\documentclass[11pt,leqno]{article}
\usepackage{amsmath,amssymb,amsfonts,graphicx}
\newtheorem{theorem}{Theorem}[section]
\newtheorem{proposition}[theorem]{Proposition}

\newtheorem{lemma}[theorem]{Lemma}
\newtheorem{definition}[theorem]{Definition}
\newtheorem{remark}[theorem]{Remark}

\newtheorem{example}[theorem]{Example}

\newenvironment{proof}{\mbox{\bf Proof.}}{\mbox{$\dashv$}\bigskip}
\begin{document}
 \begin{center}
 {\Large\bf  Epistemic Learning Programs}\\
 A Calculus  for Describing  Epistemic Action Models\\

\vspace{.25in}
{\bf Mohammad Ardeshir and Rasoul Ramezanian}\\
Department of Mathematical Sciences,\\
Sharif University of Technology,\\
P. O. Box 11365-9415, Tehran, Iran\\
mardeshir@sharif.edu, ramezanian@sharif.edu
\end{center}
\begin{abstract}
\noindent Dynamic Epistemic Logic  makes it possible to model and
reason about  information change in  multi-agent systems.
\emph{Information change} is mathematically modeled through
epistemic action Kripke models introduced by Baltag et al. Also,
van Ditmarsch interprets the information change as a
\emph{relation} between epistemic states and sets of epistemic
states and to describe it formally, he considers a special
constructor $L_B$ called learning operator. Inspired by this, it
seems natural to us that the basic source  of information change
in a multi-agent system should be \emph{learning} an announcement
by some agents \emph{together}, \emph{privately},
\emph{concurrently} or even \emph{wrongly}. Hence moving along
this path, we introduce the notion of a \emph{learning program}
and prove that all finite $K45$ action models can be described by our
learning programs.
\end{abstract}

\section{Introduction}

A computable  function over   strings of a finite alphabet is a
function that can be computed by a Turing machine. A Turing
machine takes a  string as   input, performs a sequence of
elementary changes  on the string and if it halts, it provides
another string as  output of the function. In recursion theory
 all Turing computable function can be obtained via some
\emph{initial functions}: zero, successor, and projections through
applying some \emph{basic operations} such as composition,
primitive recursion and least search. In this paper, our goal is
to develop a similar methodology for a class of epistemic
functions. Following the same terminology, an epistemic function
is a function that takes the epistemic state of a multi-agent
system as input and yields  a new epistemic state as output. The
notion of epistemic function is the focus of Dynamic Epistemic
Logic, Baltag et al.~\cite{kn:bal1,kn:bal}, and is formalized in
action models. These functions act on Kripke models via an update
operator and produce an update Kripke model. In this paper, we
concentrate on those   epistemic functions which can be coded as
$K45$ action models.  $K45$ models are those models which
accessibly relations transitive and Euclidian. We claim that there
are possible information changes which are not possible to encode
them by $KD45$ or $S5$ action models. Consequently $K45$ action
models are more powerful to describe epistemic functions than
$KD45$ and $S5$ action models. It is why we consider $K45$ action
models instead of $S5$ or $KD45$ models.

So far no one has looked at it from a computational aspect to
answer the following question
\begin{quote}
what are the \emph{initial functions} and the \emph{basic
operations} by which all $K45$ epistemic functions  can be obtained?
\end{quote}

The basic source of information change  in a multi-agent system is
\emph{learning} an announcement by some agents \emph{together},
\emph{privately}, \emph{concurrently} or even \emph{wrongly}. So
the \emph{basic operators} should be different kinds of learning.
Van Ditmarsch et al. introduced a learning constructor
in~\cite{kn:dit1,kn:dit2}. We define our own learning operator
which is different from van Ditmarsch's.

As our \emph{initial functions}, we take the \emph{test} of any
facts $\varphi$, that is $?\varphi$. For the \emph{basic
operations}, we take the following different kinds of learning: 1-
\emph{alternative learning}, 2- \emph{concurrent learning}, 3-
\emph{wrong learning}, and finally, 4- \emph{recursive learning}.
Following the footsteps of recursion theory, we prove that all
epistemic functions can be obtained through the test of facts by
applying the above four \emph{basic operations}.

Epistemic logic, started with Hintikka's
groundwork~\cite{kn:hintk}, models and reasons about the knowledge
of agents in a group~\cite{kn:reas}. In Epistemic logic, the
notions of knowledge and belief are modeled in terms of the
 possible worlds (states). An agent knows or believes a
fact if it is true in all the worlds that the agent considers
possible as alternatives for the actual world.

As information is transmitted, knowledge and belief of agents in
a multi-agent system may change. The simplest cause for an
information change is to announce some truth in public. Plaza in
1989~\cite{kn:plaz}, introduced a logic to formalize these
changes and called it the public announcement logic. In public
announcement logic, a fact is publicly announced in a multi-agent
system and all agents together update their knowledge and belief.
However, more complex actions than such as private or dishonest
announcements may occur, this is whereby  different agents may
have different views on some
action~\cite{kn:bal1,kn:bal,kn:dit1}, or the announcement may not
be truthful.

Our work may be considered as a bridge between two paradigms in
dynamic epistemic logic, namely Baltag et al. style action model
and van Ditmarsch et al. epistemic actions.

The paper is organized as follows. In section~\ref{bak}, we recall
the definition of action models from~\cite{kn:bal1}, and we
explain what it means that a group of agents learns an action
model. We discuss $K45$ Kripke models which only have transitive
and Euclidian properties. We also introduce a new notion called
the {\em applicable formulas}, and restrict the definition of
satisfaction relation only to applicable formulas.

In section~\ref{sdel}, we introduce the initial functions and
basic operations as the building blocks of the finite epistemic
functions.

In section~\ref{LBR}, we add the new operator of recursion in
constructing the finite epistemic functions.

Finally, in section~\ref{CRFW}, we compare our work with other
related works, and further works may be done.
\section{Backgrounds}\label{bak}
In this section, we introduce two significant notion: 1- learning
an action model, 2- $K45$ Kripke models and applicable formulas.

\subsection{Action models}\label{mean1}

\noindent  We start by recalling the definition of an action model
from Baltag~\cite{kn:bal1}. A pointed action model $(N,s_0)$ is a
tuple $N=\langle S,(\rightharpoonup_a)_{a\in A},pre\rangle$, where
$S$ is a set of events, $A$ is a set of agents,
$\rightharpoonup_a$ is an accessibility relation for agent $a$ on
events in $S$, and $pre$ is a function that assigns to each
event, a formula of an appropriate  epistemic language, as a
precondition for that event. $s_0\in S$ is called the {\em actual
event.}

Announcements of facts in multi-agent systems give rises to
information changes, where different agents have different access
to the resource of the announcement, and also different views
about the access of other agents to the resource.

For instance, consider the pointed action model $(N_1,s)$ in
Figure~1, where $S=\{s,t\}$, $ s\rightharpoonup_a t$,
$t\rightharpoonup_a t$, $pre(s)=\varphi, pre(t)=\psi$.
\begin{center}
\scalebox{.6}{\includegraphics{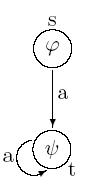}}

Figure~1
\end{center}
Here $(N_1,s)$ encodes the following information change
\begin{quote} `` $\varphi$ is announced whereas agent $a$  (wrongly)
learns $\psi$."
\end{quote}
One may assume that $\varphi=green$ and $\psi=blue$, are two
different colors. Then the action model $(N_1,s)$ says that a
green ball is shown to $a$, whereas agent $a$ thinks that she
sees a blue ball.

Consider the pointed model $(N_2,s)$ in Figure~2, where
$S=\{s,t\}$, $s\rightharpoonup_b s$, $ s\rightharpoonup_a t$,
$t\rightharpoonup_a t$, $pre(s)=\varphi$ and $pre(t)=\psi$.
\begin{center}
\scalebox{.6}{\includegraphics{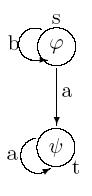}}

Figure~2
\end{center}
This action model encodes the following information change: ``a
green ball is shown to agents $a$ and $b$, agent $b$ sees a green
ball and is aware that agent $a$ has a color-blindness and sees a
blue ball. Agent $a$ just sees a blue ball and has no \emph{idea}
about what $b$ sees".

\begin{remark}
{\em The word ``\emph{having no idea}" used above is vague, and
needs to be clarified. In the action model $N_1$,  agent $b$ is
not present in the state $t$. So agent $a$ has no idea about the
information of agent $b$. There could be lots of possibilities
about the color-blindness of $b$ at  state $t$, but the action
model says nothing about it. We will later introduce the notion
of \emph{applicable formulas}~(see Definition~\ref{applic}) to
formally model this case.}
\end{remark}

For another example, consider the pointed model $(N_3,s)$ in
Figure~3, where $S=\{s,t\}$, $s\rightharpoonup_b s$,
$t\rightharpoonup_b s$ $ s\rightharpoonup_a t$,
$t\rightharpoonup_a t$, $pre(s)=\varphi$ and $pre(t)=\psi$.
\begin{center}
\scalebox{.6}{\includegraphics{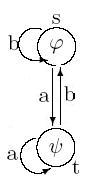}}

Figure~3
\end{center}
The action model $(N_3,s)$ encodes the information change that ``a
green ball is shown to agents $a$ and $b$,  agent $b$ sees a
green ball and is aware that agent $a$ has a color-blindness and
sees a blue ball. Agent $a$ sees a blue ball and wrongly thinks
that agent $b$ has a color-blindness and sees a green ball.
Moreover, both agents are aware about each other's thoughts".

According to the above discussion, a pointed action model $(N,s)$,
where $N=\langle S,(\rightharpoonup_a)_{a\in A},pre\rangle$,
encodes the information change that

\begin{quote}
\emph{the fact $pre(s)$ is announced whereas each agent (relevant
to its accessibility relation in the action model) acquires
information of what may have been announced and what other agents
may have heard.}
\end{quote}


\subsection{ Learning an Action  Model}\label{mean2}
We should also clarify what it means that a group of agents learns
an \emph{action model}. Suppose  $(N,s)$ is a pointed action
model, and $B$ is a group of agents. The case that
\begin{quote}
``group $B$ learns $(N,s)$",
\end{quote}
is a new action model which encodes the following information
change
\begin{quote}
``the fact $pre(s)$ is announced and group $B$ learns the fact
$pre(s)$ and the information that other agents (excluding $B$)
acquires due to learning $(N,s)$".
\end{quote}

For example, recall the pointed action model $(N_1,s)$ in
Figure~1, where it encodes the following information change

\begin{quote}
``$\varphi$ is announced whereas agent $a$ (wrongly) learns
$\psi$".
\end{quote}
An agent $b$ learns $(N_1,s)$ means
\begin{quote}
``$\varphi$ is announced and agent $b$ learns that $\varphi$ is
announced  and the fact that agent $a$ (wrongly) learns $\psi$.
Agent $a$ still wrongly learns $\psi$ and has no idea about what
$b$ learns".
\end{quote}
The case that agent $b$ learns $(N_1,s)$, denoted by
$L_b((N_1,s))$, is encoded by the pointed action model $(N_2,s)$
in Figure~2.

Four different kinds of learning, in a multi-agent system, may be
distinguished:

\begin{itemize} \item[1.] \textsc{Alternative Learning}: a group of agents,
$B$, together learns that among a set of actions
$(M_1,s_1),(M_2,s_2),..., (M_k,s_k)$, one of them is the actual
action.

\begin{example}{\em
Let $(M_1,s_1)$ be the pointed action model $(N_1, s_1)$ in
Figure~1, and $(M_2,s_2)$ be another action model, in which
$M_2=\langle S=\{s_2\},\emptyset,pre(s_2)=\chi\rangle$. Then the
pointed action model of $L_b((M_1,s_1),(M_2,s_2))$ is
\begin{center}
\scalebox{.6}{\includegraphics{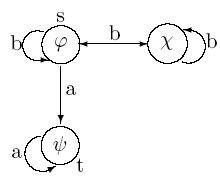}}

Figure~4
\end{center}
See Definition~\ref{defsem} for details.
 }
\end{example}

\item[2.] \textsc{Concurrent Learning}: two disjoint groups of agents
$B_1$ and $B_2$ learn concurrently but not together; group $B_1$
learns $(N_1,s_1)$ and group $B_2$ learns $(N_2,s_2)$ concurrently
but not together.

\begin{example}{\em Let $(M_1,s_1)$ be the pointed action model
$(N_1, s_1)$ in Figure~1, and $(M_2,s_2)$ be another pointed
action model, in which $M_2=\langle
S=\{s_2\},(s_2\rightharpoonup_b s_2) ,pre(s_2)=\varphi\rangle$.
Then the pointed action model of $(M_1,s_1)\cap(M_2,s_2)$ would be
\begin{center}
\scalebox{.6}{\includegraphics{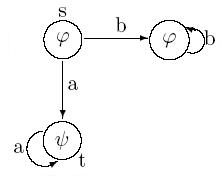}}

Figure~5
\end{center}
See Definition~\ref{defsem} for details.
 }
\end{example}

\item[3.] \textsc{Wrong Learning}: whereas a fact $\psi$ is
announced, a group of agents, $B$, wrongly learns something else.

\begin{example}{\em
$\varphi$ is announced and agent $a$ learns $\psi$. See the
pointed action model $(N_1,s_1)$ in Figure~1.}
\end{example}

\item[4.] \textsc{Recursive Learning}: a group of agents, $B_1$,
learns what another group of agents, $B_2$, learns, and group
$B_2$ learns what group $B_1$ learns.
\begin{example}\label{rexm}{\em
Consider the pointed action model $(N_3,s)$ in Figure~3.
\begin{center}
\scalebox{.6}{\includegraphics{fig3.jpg}}

Figure~3
\end{center}
The action model $(N_3,s)$ means

\noindent  $\varphi$ is announced and agent $b$ learns

\begin{quote}$\{$
$\varphi$ and the case that agent $a$ (wrongly) learns
\begin{quote}
$\{\psi$ and about what $b$ (wrongly) learns$\}$
\end{quote}
$\}$.
\end{quote}
}
\end{example}
\end{itemize}

\subsection{Epistemic Logic}

In this section we briefly go through the syntax and semantics of
epistemic logic. The syntax of epistemic logic is as usual, but
the semantics is a little bit different from the standard one.

\begin{definition} Let $P$ be a non-empty set of propositional
variables, and $A$ be a set of agents. The language $L(A,P)$ is
the smallest superset of $P$ such that
\begin{center}
if $\varphi,\psi\in L(A,P)$ then $\neg \varphi,\
(\varphi\wedge\psi),\ K_i\varphi\in L(A,P)$,
\end{center}
for $i\in A$.

For $i\in A$, $K_i\varphi$ has to be read as `agent $i$ believes
(knows) $\varphi$". For a group of agents $B\subseteq A$,
$K_B\varphi$ means that $K_i\varphi$, for all $i\in B$.
\end{definition}

Epistemic logic models the notions of knowledge and belief in
terms of the notion of possible worlds in Kripke semantics.

\begin{definition}\label{applic}
A Kripke model $M$ is a tuple $M=\langle
S,(\rightharpoonup_i)_{i\in A}, V\rangle$, where $S$ is a
non-empty set of worlds \emph{(states)} $s\in S$, $V$ is a
function from $P$ to $2^S$, and each $\rightharpoonup_i$ is a
binary accessibility relation between worlds. We define the group
\emph{present} at the state $(M,s)$ as follows:
\begin{center}
$gr((M,s))=\{i\in A\mid (\exists t\in S) ~s\rightharpoonup_i
t\}$.
\end{center}
For an epistemic state $(M,s)$, the set of \emph{applicable}
formulas at the state $(M,s)$, denoted by $\Phi_{(M,s)}\subseteq
L(A,P)$, is the smallest subset satisfying the following
conditions

\begin{itemize}
\item[1.] $P\subseteq \Phi_{(M,s)}$,

\item[2.] if $\varphi,\psi\in \Phi_{(M,s)}$ then
$\varphi\wedge\psi\in \Phi_{(M,s)}$ and $\neg\varphi\in
\Phi_{(M,s)}$,

\item[3.] $K_i\varphi\in \Phi_{(M,s)}$ if and only if  $i\in
gr((M,s))$ and for all $t$ such that $s\rightharpoonup_i t$,
$\varphi\in \Phi_{(M,t)}$.

\end{itemize}
\end{definition}

Intuitively, the \emph{applicable} formulas of an epistemic state,
are those formulas that can sensibly be assigned a truth value.
For example, consider the Kripke model in Figure~6.
\begin{center}
\scalebox{.6}{\includegraphics{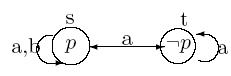}}

Figure~6
\end{center}
As agent $b$ is not present in the world $t$,  formulas like $K_b
\chi$ are not applicable in the world $t$ (it is not possible to
talk about the truth of $K_b \chi$ in world $t$, where agent $b$
is not present in this world). Also the formula $K_a(K_b p\vee
K_b\neg p)$ is not applicable at the world $s$. In the next
definition, we restrict  the definition of truth to applicable
formulas.

\begin{definition}
In order to determine whether an applicable formula $\varphi\in
\Phi_{(M,s)}$ is true in the epistemic state $(M,s)$, denoted by
$(M,s)\models \varphi$, we look at the structure of $\varphi$:
\[
\begin{array}{l}\begin{array}{ccccc}
  (M,s)\models p & \emph{iff} &s\in V(p) \\
  (M,s)\models
(\varphi\wedge\psi) & \emph{iff} & (M,s)\models\varphi~and~(M,s)\models\psi \\
  (M,s)\models\neg\varphi & \emph{iff} & not~(M,s)\models\varphi\ ((M,s)\not\models\varphi) \\
  (M,s)\models
K_i\varphi & \emph{iff} & \emph{for~all}~t~\emph{such~that}~
s\rightharpoonup_it,~(M,t)\models\varphi \\

\end{array}\quad\quad\quad\quad\quad\quad\quad\quad\quad\quad\quad\quad\quad
\end{array} \]

\end{definition}
Note that the satisfaction relation is just defined for
\emph{applicable} formulas.

The standard epistemic logic $S5$ consists of axioms $A1-A5$ and
the derivation rules $R1$ and $R2$ given below
\[
\begin{array}{l}\emph{R1:~}
\vdash\varphi,\ \vdash\varphi\rightarrow\psi\Rightarrow\ \vdash\psi\\
\emph{R2:~}\vdash\varphi\Rightarrow
K_i\varphi,\emph{~for~all}~i\in A
\quad\quad\quad\quad\quad\quad\quad\quad\quad\quad\quad\quad\quad\quad\quad\quad\quad\quad\quad
\end{array} \]
\[
\begin{array}{l}\emph{A1:~ Axioms~of~propositional~logic}\quad\quad\quad\quad\quad
\quad\quad\quad\quad\quad\quad\quad\quad\quad\quad\quad\quad\\
\emph{A2:~} (K_i\varphi\wedge
K_i(\varphi\rightarrow\psi))\rightarrow K_i\psi\\
\emph{A3:~} K_i\varphi\rightarrow\varphi\\
\emph{A4:~} K_i\varphi\rightarrow K_iK_i\varphi\\
\emph{A5:~} \neg K_i\varphi\rightarrow K_i\neg K_i\varphi\\
\end{array} \] If instead of $A3$, we assume  the weaker axiom $D$ (given below),
the logic of belief $KD45$ will be  specified.\[
\begin{array}{l}\emph{D:~}  \neg(K_i\varphi\wedge K_i\neg\varphi).  \quad\quad\quad\quad\quad
\quad\quad\quad\quad\quad\quad\quad\quad\quad\quad\quad\quad\quad\quad\quad\quad\quad\quad \\
\end{array} \]

\begin{definition}\label{kddef}
Let $M=\langle S,(\rightharpoonup_a)_{a\in A}, V\rangle$ be a
Kripke model. For each $a\in A$, we say that the relation
$\rightharpoonup_a$ is
\begin{itemize}

\item[$1)$] \emph{reflexive} if and only if for all $s\in S$,
$s\rightharpoonup_as$;

\item[$2)$] \emph{serial} if and only if for all $s\in S$,  there
exists $t\in S$ such that $s\rightharpoonup_at$;

\item[$3)$] \emph{transitive} if and only if for all $s,t,u\in S$,
if $s\rightharpoonup_at$ and $t\rightharpoonup_au$ then
$s\rightharpoonup_au$;

\item[$4)$] \emph{Euclidean} if and only if  for all three
 states $s,t,u\in S$
 if $s\rightharpoonup_at$ and $s\rightharpoonup_au$ then
$t\rightharpoonup_au$.
\end{itemize}
We say a relation $\rightharpoonup_a$ is

\begin{itemize}\item $S5$ whenever it is reflexive, transitive and
Euclidean, \item $KD45$ whenever it is \emph{serial}, transitive
and Euclidean, and

\item $K45$ whenever it is transitive and
Euclidean.\end{itemize}
\end{definition}
A Kripke model $M$ is called an $S5$ model if and only if for each
agent $a$, $\rightharpoonup_a$ is $S5$. It is obvious that every
$S5$ model is a model of the {\em standard} epistemic logic. A
Kripke model $M$ is called an $KD45$ model if and only if for each
agent $a$, $\rightharpoonup_a$ is $KD45$.

We also say a Kripke model $M$ to be a $K45$ model if and only if
for each agent $a$, $\rightharpoonup_a$ is $K45$.

An action model $\mathrm{N}$ is called a $K45$ model if and only
if for each agent $a$, $\rightharpoonup_a$ is $K45$, and it is an
$S5$ model if and only if for each agent $a$, $\rightharpoonup_a$
is $S5$. For example, the action model in Figure~2 is $K45$.
\begin{definition}
Let $P$ be a set of atomic formulas, $A$ be a set of agents, and
$\Phi$ be a set of epistemic formulas over atomic formulas $P$ and
agents in $A$.

\begin{itemize}
\item[] We use $\mathrm{FAct}(\Phi)$ (or simply $\mathrm{FAct}$)
to refer to the set of all pointed \emph{epistemic action} models
$(\mathrm{N},t)$ such that $N$ is a \emph{finite} $K45$ model, and
the image of $pre$ is $\Phi$.

\item[] We use $\mathrm{Mod}(A,P)$ (or simply $\mathrm{Mod}$) to
refer to the set of all \emph{epistemic states} $(M,s)$ such that
$M$ is $K45$.

\end{itemize}
\end{definition}

We recall definition of bisimulation of actions
from~\cite{kn:dit3}. Consider two action models
$\mathrm{N}=\langle S,(\rightharpoonup_a)_{a\in A}, pre\rangle$
and $\mathrm{N}'=\langle S',(\rightharpoonup'_a)_{a\in A},
pre'\rangle$. The pointed action model $(\mathrm{N},s)$ is
bisimilar to $(\mathrm{N}',s')$, denoted by $(\mathrm{N},s)\simeq
(\mathrm{N}',s')$, whenever there is a relation
$\mathcal{R}\subseteq S\times S'$   satisfying the following
conditions, for each agent $a\in A$:
\begin{itemize}
\item[]Initial. $\mathcal{R}(s,s')$.

\item[] Forth. If $\mathcal{R}(t,t')$ and $t\rightharpoonup_a v$,
then there is a $v'\in S'$ such that $\mathcal{R}(v,v')$ and
$t'\rightharpoonup'_a v'$.

\item[]Back. If $\mathcal{R}(t,t')$ and $t'\rightharpoonup'_a v'$,
then there is a $v\in S$ such that $\mathcal{R}(v,v')$ and
$t\rightharpoonup_a v$.

\item[] Pre. If $\mathcal{R}(t,t')$, then $pre(t)$ is equivalent
to $pre(t')$ in $KD45$ belief logic.
\end{itemize}

We define another notion of equivalence on action models and call
it {\em agent-bisimulation.}
\begin{definition}
Let $a\in A$ be an arbitrary agent. Two pointed action models
$(\mathrm{N},s)$ and $(\mathrm{N}',s')$ are $a$-bisimilar whenever
\begin{itemize}

\item[] Forth. If $s\rightharpoonup_a t$,
then there is a $t'\in S'$ such that $s'\rightharpoonup'_a t'$
and $(\mathrm{N},t)\simeq (\mathrm{N}',t')$.

\item[]Back. If $s'\rightharpoonup'_a t'$,
then there is a $t\in S$ such that $s\rightharpoonup_a t$ and
$(\mathrm{N},t)\simeq (\mathrm{N}',t')$.

\end{itemize}
\end{definition}

The execution of a pointed action model $(\mathrm{N},t)\in
\mathrm{FAct}$ on an epistemic state $(M,s)\in \mathrm{Mod}$ is a
new epistemic state $(M\ast\mathrm{N}, (s,t))$, where
$M\ast\mathrm{N}=\langle S,(\rightharpoonup_i)_{i\in A}, V\rangle$
and
\begin{itemize}
\item[] $S=\{(s_1,t_1)\mid s_1\in \mathrm{N}, t_1\in M,
(M,s_1)\models pre(t_1)\}$,

\item[] $(s_1,t_1)\rightharpoonup_i (s_2,t_2)$ iff
$s_1\rightharpoonup_i s_2$ and $t_1\rightharpoonup_i t_2$,

\item[]$(s_1,t_1)\in V(p)$ iff $s_1\in V(p)$.

\end{itemize}

\begin{proposition}
Suppose $(\mathrm{N},t)\in \mathrm{FAct}$ and $(M,s)\in
\mathrm{Mod}$. Then $(M\ast N, (s,t))\in \mathrm{Mod}$, i.e., it
is a $K45$ model.
\end{proposition}
\begin{proof} It is straightforward.
\end{proof}


\subsection{Why $K45$  Models and Applicable Formulas?}
$K45$ Kripke models are more general than $KD45$ models, as they
necessarily do not have the serial property. The reason that we
consider serial property for $KD45$ model is that, we want the
agent's belief to be consistent. For $K45$ Kripke models, we
consider the definition of satisfaction relation just for {\em
applicable} formulas. In this way, the agent's beliefs are
consistent at each state. Moreover,

\begin{quote}regarding {\em
applicable} formulas, the class of $K45$ models  is a (sound and
complete) semantics for   logic of belief $KD45$.
\end{quote}
 Assume that a formula $\varphi$ is
derivable from the logic of belief, i.e., $KD45\vdash \varphi$.
Then for every $K45$ pointed model $(M,s)$, if $\varphi$ is
applicable at this state, then $(M,s)\models \varphi$. Also, note
that every $KD45$ model is also a $K45$ model. So for any formula
$\varphi$, if for all $K45$ pointed model $(M,s)$ which $\varphi$
is applicable at the state, we have $(M,s)\models \varphi$ then
$KD45\vdash \varphi$.

Since $K45$ models are more general than $KD45$ models, we can
encode more epistemic functions in $K45$ models, as we do not
have to determine the cases where an agent does not have any idea
about the belief of another agent.  For example, consider the
following information change:

\begin{quote}
a green ball is shown to agents $a$ and $b$, agent $b$ sees a
green ball and is aware that agent $a$ has a color-blindness and
sees a blue ball. Agent $a$ just sees a blue ball and has no idea
about what $b$ sees.
\end{quote}
In the above information change, agent $a$ has no idea about what
$b$ sees. Agent $a$ says I do not have any idea, there could be
lots of possibilities and I do not know even how many
possibilities exist, may be an infinite number of them. It would
be possible that agent $b$ sees a cube instead of a ball, or even
an elephant, and etc. It would be possible that agent $b$ sees the
ball in a color which is unknown for me.

In $KD45$ models, we are forced to encode all possibilities, but
if it happens that some possibilities are unknown, then we don't
know what to do. However such  information change can   be encoded
by the $K45$ action model in figure~2.

Therefore, considering applicable formulas, the class of $K45$
models is still a semantics for logic of belief (similar to the
class of $KD45$ models) and moreover, they are enough general than
$KD45$ models for describing information changes formally.

\section{Basic Learning Programs}\label{sdel}
As mentioned in the introduction, our goal is to introduce some
initial functions and  basic operators as the building blocks of
all finite epistemic functions ($\mathrm{FAct}$). In this
section, we introduce the first class of epistemic learning
programs, called {\em basic learning programs}. The reason these are
called {\em basic} is that they do not include any recursion in their
structure.

\begin{definition}
Let $\Phi$ be a set of epistemic formulas over a set of atomic
formulas $P$ and a set of agents $A$.  The set of \emph{basic
learning programs} $\mathrm{BLP}(\Phi)$ is defined as follows:

\begin{itemize}

\item[i.]  \verb"Test". for all $\varphi\in \Phi$, $?\varphi$ is a
basic learning program, and we define $group(?\varphi)=\emptyset$,
and $pre(?\varphi)=\varphi$,

\item[ii.]\verb"Alternative Learning". for all $n\in \mathbb{N}$, and
$B\subseteq A$, if $\alpha_1,\alpha_2,...,\alpha_n$ are basic
learning programs, then $L_B(\alpha_1,\alpha_2,...,\alpha_n)$ is a
basic learning program, and we define
$group(L_B(\alpha_1,\alpha_2,...,\alpha_n))=B\cup
group(\alpha_1)$ and
$pre(L_B(\alpha_1,\alpha_2,...,\alpha_n))=pre(\alpha_1)$,

\item[iii.]\verb"Concurrent Learning". if $\alpha_1,\alpha_2$ are
basic learning programs such that $pre(\alpha_1)=pre(\alpha_2)$
and $group(\alpha_1)\cap group(\alpha_2)=\emptyset$, then
$\alpha_1\cap\alpha_2$ is a basic learning program. We define
$group(\alpha_1\cap\alpha_2)=group(\alpha_1)\cup group(\alpha_2)$,
and $pre(\alpha_1\cap\alpha_2)=pre(\alpha_1)$,

\item[iv.]\verb"Wrong Learning". if $\alpha_1$ is a basic learning
program and $\psi$ is an epistemic formula, then $\psi|_B\alpha_1$
is a basic learning program whenever $B\subseteq group(\alpha_1)$.
We define $group(\psi|_B\alpha_1)=B$ and
$pre(\psi|_B\alpha_1)=\psi$.
\end{itemize}
\end{definition}

To each basic learning program, we associate a pointed action
model as follows.

\begin{definition}\label{defsem}\textbf{\emph{Semantics of BLP}}.
\begin{itemize}

\item[1.] for all $\varphi\in \Phi$, the pointed action model of
the program $?\varphi$,  is
$(\mathrm{N}_{?\varphi},s_{?\varphi})$, where
\begin{center}
$\mathrm{N}_{?\varphi}=\langle
\{s_{?\varphi}\},(\rightharpoonup_a)_{a\in A}, pre\rangle$,
\end{center}
in which for all $a\in A$, $\rightharpoonup_a=\emptyset$ and
$pre(s_{?\varphi})=\varphi$.

\item[2.]  Suppose $\alpha_1,\alpha_2,...,\alpha_k$ are basic
learning programs and their associated action models are
$(N_1,s_1), (N_2,s_2), ..., (N_k,s_k)$, where $N_l=\langle
S_l,(\rightharpoonup^l_a)_{a\in A}, pre_l\rangle$, for $1\leq
l\leq k$, and $\cap_{1\leq l\leq k}S_k=\emptyset$. Then the
associated action model to the basic learning program
$L_B(\alpha_1,\alpha_2,...,\alpha_n)$ is $(N,s)$, where $N=\langle
S,(\rightharpoonup_a)_{a\in A}, pre\rangle$ and
\begin{itemize}
\item $S=\{(s_1,1),(s_2,1),...,(s_k,1)\}\cup S_1\cup
S_2\cup...\cup S_k$,

\item for all $b\in B$, for all $1\leq i,j\leq k$, if $s_i$ is
$b$-bisimilar to $s_j$ then $(s_i,1)\rightharpoonup_b (s_j,1)$,

\item for all $1\leq j\leq k$, for all $a\in gr(N_j,s_j)-B$, for
all $t\in S_j$, if $s_j\rightharpoonup^j_a t$ then
$(s_j,1)\rightharpoonup_a t$,

\item for all $a\in A$, for all $1\leq i\leq k$,  for all $v,t\in
S_i$, if $v\rightharpoonup^i_a t$ then $v\rightharpoonup_a t$,

\item $s=(s_1,1)$,

\item $pre((s_l,1))=pre(s_l)$ for all $1\leq l\leq k$.
\end{itemize}
The associated action model of
$L_B(\alpha_1,\alpha_2,...,\alpha_n)$ means:   $pre(\alpha_1)$ is
announced whereas  agents in $B$ learn  that either $\alpha_1$ or
$\alpha_2$ or ... or $\alpha_k$ has been executed, and what other
agents learn in each alternative case.

\item[3.] Suppose $\alpha_1,\alpha_2$ are basic learning programs
and their associated action models are $(N_1,s_1)$ and $(N_2,s_2)$
respectively, where $N_l=\langle S_l,(\rightharpoonup^l_a)_{a\in
A_l}, pre_l\rangle$ for $l=1$ or $l=2$, $A_1\cap A_2=\emptyset$
and $S_1\cap S_2=\emptyset$. Then the pointed action model of the
basic learning program $\alpha_1\cap\alpha_2$ is $(N,s)$ with
$N=\langle S,(\rightharpoonup_a)_{a\in A}, pre\rangle$, where
\begin{itemize}
\item $S=\{s\}\cup S_1\cup S_2$, for some $s\not \in S_1\cup S_2$,

\item for all $1\leq i\leq 2$, for all $a\in A_i$,  if
$v\rightharpoonup^i_a t$ then $v\rightharpoonup_a t$,

\item for all $1\leq j\leq 2$, for all $a\in A_j$, if
$s_j\rightharpoonup^j_a t$ then $s\rightharpoonup_a t$,

\item $pre (s)= pre(s_1)$.
\end{itemize}
The associated action model of $\alpha_1\cap\alpha_2$ means:
$pre(\alpha_1)$ is announced whereas agents in $group(\alpha_1)$
learn according to execution of $\alpha_1$, agents in
$group(\alpha_2)$ learn according to execution of $\alpha_2$.

\item[4.]

Suppose $\alpha$ is a basic learning program and its associated
action model is $(N_1,s_1)$, where $N_1=\langle
S,(\rightharpoonup^1_a)_{a\in A}, pre_1\rangle$. Then the
associated action model of $\psi |_B\alpha$ is $(N,s)$ with
$N=\langle S,(\rightharpoonup_a)_{a\in A}, pre\rangle$, where
\begin{itemize}
\item $S=\{s\}\cup S_1$, for some $s\not \in S_1$,

\item for all $a\in A$, if $v\rightharpoonup^1_a t$ then
$v\rightharpoonup_a t$,

\item for all $b\in B$, if $s_1\rightharpoonup^1_b t$ then
$s\rightharpoonup_b t$,

\item $pre (s)= \psi$.
\end{itemize}

The associated action model of $\psi|_B\alpha_1$ means:   $\psi$
is announced whereas agents in group $B$ wrongly learn according
to execution of $\alpha_1$.
\end{itemize}
\end{definition}

The action model of a program $\alpha$ is denoted by
$(N_\alpha,s_\alpha)$. An \emph{epistemic action}
$(\mathrm{N},t)\in \mathrm{FAct}(\Phi)$ is called \emph{basic
learning action} whenever there is a basic learning programs
$\alpha\in\mathrm{BLP}(\Phi)$  such that $(N_\alpha,s_\alpha)$ is
bisimilar to $(\mathrm{N},t)$.

\begin{example}{\em The action model associated with the basic
learning program
\begin{center}
$L_b(\varphi|_aL_a?\psi)$.
\end{center} is bisimilar to the action model illustrated in
Figure~2.

\noindent The action model illustrated in Figure~4 is a basic
learning action. It is bisimilar to the action model associated
with the basic learning program
\begin{center}
$L_b(\varphi|_aL_a?\psi,?\chi)$.
\end{center}
Also, the action model associated with  the basic learning program
\begin{center}
$\varphi|_aL_a?\psi\cap \varphi|_bL_b?\varphi$.
\end{center}
is bisimilar with  the action model illustrated in Figure~5. }
\end{example}

\begin{definition}\label{closef}
Let $N=\langle S,(\rightharpoonup_a)_{a\in A}, pre\rangle$ be an
$K45$ action model.
\begin{itemize}

\item An $S5$ submodel of $N$ is an $S5$ action model $M=\langle
S',(\rightharpoonup'_a)_{a\in B}, pre'\rangle$, where $S'\subseteq
S$, $pre'=pre|_{S'}$, $B\subseteq A$ and for all $a\in B$,
$\rightharpoonup'_a= \rightharpoonup_a|_{S'}$. An $S5$ submodel
is called \emph{connected} whenever all two different states of
the model are \emph{reachable} from each other. It is called
\emph{closed} whenever for all $s,t\in S$, if $s\in S'$ and for
some $a\in B$, $s\rightharpoonup_a t$, then $t\in S'$.

\item Let $M'=\langle S',(\rightharpoonup'_a)_{a\in B}, pre'\rangle$ and
$M''=\langle S'',(\rightharpoonup''_a)_{a\in C}, pre''\rangle$ be
two closed connected $S5$ action submodels of $N$. We write
$M'\leq M''$ whenever $S'\subseteq S''$ and $B\subseteq C$. One
may easily verify that $\leq$ is a partial order relation.

\item  Assume $M^i=\langle S^i,(\rightharpoonup^i_a)_{a\in B^i},
pre^i\rangle$, $1\leq i\leq k$ are all different maximal closed
connected $S5$ submodel of $N$ with respect to the partial order
relation $\leq$. We construct the action model $T(N)=\langle
S',(\rightharpoonup'_a)_{a\in A}, pre'\rangle$ as follows:
\begin{itemize}
\item $S'=\bigcup_{1\leq i\leq k} (S^i\times\{i\})$

\item $(s,i)\rightharpoonup'_a (t,j)$ if and only if either $i=j$
and $s\rightharpoonup^i_a t$, or $i\neq j$, $a\not \in B^i$, $a\in
B^j$, and $s\rightharpoonup_a t$ (in the action model $N$).

\item $pre'((s,i))=pre^i(s)$.
\end{itemize}
\item The \emph{projection} of a state $(s,i)$ in the action
model $T(N)$, denoted by $\Pi((s,i))$, is defined to be the state
$s$ in the action model $N$.

\item For the action model $N$, we define the \emph{directed graph} of
$N$, denoted by $G(N)$ as follows:
\begin{itemize}
\item Each node $G(N)$ is a maximal closed connected $S5$
submodel $M^i$ of $N$.

\item Between two different nodes $M^i,M^j$, there exists a
directed edge $M^i\rightarrowtail M^j$ if and only if there
exists an accessibility relation in the action model $T(N)$,
$(s,i)\rightharpoonup'_a (t,j)$, for some $s\in S_i$, $a\in A$,
and $t\in S_j$.
\end{itemize}
\end{itemize}
\end{definition}

It is obvious that an $S5$ submodel of a $K45$ action model may be
connected but not closed and vice versa.

\begin{proposition}\label{ddd}
Let $N=\langle S,(\rightharpoonup_a)_{a\in A}, pre\rangle$ be a
$K45$ action model. Then $(N,s_0)$ is bisimilar to $(T(N),w_0)$
for all $s_0$ and $w_0$ for which the state $s_0$ is the
projection of $w_0$.
\end{proposition}\begin{proof} See the Appendix.\end{proof}



\begin{theorem}\label{prtree1}
If a pointed $K45$ action model $(N,s)$ is a basic learning
action, then its graph $G(N)$ is a tree.
\end{theorem}\begin{proof} See the Appendix.\end{proof}

\begin{example}{\em Consider the action model $(N_3,s)$ shown
in Figure~3. There is no basic learning program $\alpha$ such that
its action model is $(N_3,s)$. That is because $G(N_3)$ is not a
tree.
 }
\end{example}

We showed that the graph of the basic learning actions are trees.
The converse is also true, that is, for each $K45$ pointed action
model $(N,s)$, if its graph is a tree then it is associated to a
basic learning program up to bisimilarity.

\begin{proposition}\label{s5prop}
All $S5$ pointed action models are basic learning actions.
\end{proposition}\begin{proof} See the Appendix.\end{proof}





\begin{theorem}\label{prtree2}
If the graph of a finite $K45$ pointed model $(N,s)$ is a
\emph{finite} tree, then $(N,s)$ is a basic learning action.
\end{theorem}\begin{proof} See the Appendix.\end{proof}

\subsection{Comparing two Learning Operators}
In this part, we aim to compare our proposed learning operator
with the learning operator introduced in~\cite{kn:dit2}. We begin
with     the same example ``Lecture or Amsterdam" discussed
in~\cite{kn:dit2}.

\begin{quote} \emph{Anne and Bert are in a bar, sitting at a table.  A messenger  comes in
and delivers a letter addressed to Anne. The letter contains
either an invitation for a night  out in Amsterdam or an
obligation to give a lecture instead. Anne and Bert commonly know
that these are the only alternatives.}
\end{quote}

Consider the following information change scenario:

\begin{itemize}
\item (\textbf{spy-seeing}). Bert says good bye to Anne and leaves
the bar. During his leaving, he \emph{secretly} spies from the
window of the bar that whether Anne reads the letter, Anne does
not get aware that Bert spies on her, and wrongly thinks that she
is alone (Bert is not present) while she reads the letter.
\end{itemize}

 Suppose that  $p$
stands for ``Anna is invited for a night out in Amsterdam", and
also assume that in fact $p$ is true.

It is not possible to model the above information change scenario
in concurrent dynamic epistemic logic~\cite{kn:dit2}. But we can
model above scenario by the following basic learning programs.
\begin{center}
$L_b(L_a?p,L_a?\neg p)$.
\end{center}
The pointed action model associated to the above learning program
is $(N,s)$, where $S=\{s,t,v,u\}$,  $s\rightharpoonup_a v$,
$v\rightharpoonup_a v$, $t\rightharpoonup_a u$,
$u\rightharpoonup_a u$, and $s\rightharpoonup_b s$,
$s\rightharpoonup_b t$, $t\rightharpoonup_b s$,
$t\rightharpoonup_b t$, and  $pre(s)=pre(v)=p$,
$pre(t)=pre(u)=\neg p$.
\begin{center}
\scalebox{.6}{\includegraphics{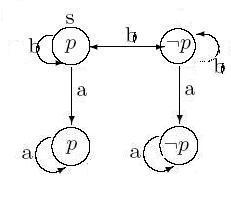}}

{\footnotesize{The pointed action model associated to
$L_b(L_a?p,L_a?\neg p)$}}
\end{center}
A candidate to describe the action ``spy-seeing" in formalization
presented in~\cite{kn:dit2} could be  $L_b( ! L_a?p\cup L_a?\neg
p)$. But the term $L_b( ! L_a?p\cup L_a?\neg p)$ is not a
well-formed action in concurrent dynamic epistemic logic since it
does not satisfy definition~6 of~\cite{kn:dit2}. In this
definition if $L_B\alpha$ is a well-formed action then
$gr(\alpha)\subseteq B$. Whereas, $gr(L_a?p\cup L_a?\neg
p)=\{a\}$, and $\{a\}\not\subseteq \{b\}$.

\noindent For another example, we discuss  the following
information change scenario:
\begin{itemize}
\item (\textbf{spy-reading}). Bert says good bye to Anne and
leaves the bar. using a hidden camera, Bert spies on   Anne when
she reads the letter, and Bert gets aware of the contents of the
letter using the camera. Anne wrongly thinks that she is alone,
and there is no spy on her.
\end{itemize}
Again, it is not possible to model the above information change
scenario in concurrent dynamic epistemic logic~\cite{kn:dit2}. But
we can model above scenario by the following basic learning
programs $L_b(L_a?p)$.

\begin{center}
\scalebox{.6}{\includegraphics{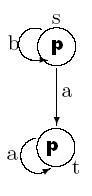}}

{\footnotesize{The pointed action model associated to
$L_b(L_a?p)$}}
\end{center}
The pointed action model associated to the above learning program
is $(N,s)$, where $S=\{s,t\}$, $s\rightharpoonup_b s$,
$s\rightharpoonup_a t$, $t\rightharpoonup_a t$, and
$pre(s)=pre(t)=p$.

Again note that the term $L_b(L_a?p)$ is not a well-formed action
in concurrent dynamic epistemic logic as $\{a\}\not\subseteq
\{b\}$.

We can also model scenarios introduced in~\cite{kn:dit2} via basic
learning programs as follows:

\begin{itemize}
\item[1-] (\textbf{tell}). Anne read the letter aloud: $L_{ab}?p$.
\begin{center}\scalebox{.6}{\includegraphics{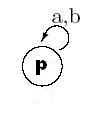}}

{\footnotesize{The pointed action model associated to $L_{ab}?p$}}
\end{center}

\item[2-] (\textbf{read}). Both Bert and Anne are present in the
bar, sitting on a table, and Bert is seeing that Anne reads the
letter: $L_{ab}(L_a?p, L_a?\neg p)$.
\begin{center}
\scalebox{.6}{\includegraphics{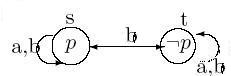}}

{\footnotesize{The pointed action model associated to
$L_{ab}(L_a?p, L_a?\neg p)$}}
\end{center}
\item[3-] (\textbf{mayread}) Bert orders a drink  at the bar so
that Anne may have read the letter (and actually Anne reads the
letter)

\begin{center}
$L_{ab}(\alpha,\beta,\gamma)$,
\end{center}
where \begin{itemize}

 \item $\alpha:= L_a?p\cap p|_b L_b?\top$ (Anne reads the letter, and in fact $p$ is true),

\item $\beta:= L_a?\neg p\cap \neg p|_b L_b?\top$ (Anne reads the
letter, and in fact $\neg p$ is true),

\item $\gamma:= L_a?\top\cap   L_b?\top$ (Anne does not read the
letter)
\end{itemize}
\begin{itemize}\item[-]
The pointed action model associated to $\alpha$, denoted by
$(N_1,s_1)$, is $S=\{s_1,t_1,v_1\}$, $s_1\rightharpoonup_a t_1$,
$s_1\rightharpoonup_b v_1$, $t_1\rightharpoonup_a t_1$,
$v_1\rightharpoonup_b v_1$, and $pre(s_1)=pre(t_1)=p$,
$pre(v_1)=\top$.

\item[-] The pointed action model associated to $\beta$, denoted
by $(N_2,s_2)$, is $S=\{s_2,t_2,v_2\}$, $s_2\rightharpoonup_a
t_2$, $s_2\rightharpoonup_b v_2$, $t_2\rightharpoonup_a t_2$,
$v_2\rightharpoonup_b v_2$, and $pre(s_2)=pre(t_2)=\neg p$,
$pre(v_2)=\top$.

\item[] The pointed action model associated to $\beta$, denoted by
$(N_3,s_3)$, is $S=\{s_3,t_3,v_3\}$, $s_3\rightharpoonup_a t_3$,
$s_3\rightharpoonup_b v_3$, $t_3\rightharpoonup_a t_3$,
$v_3\rightharpoonup_b v_3$, and $pre(s_3)=pre(t_3)=pre(v_3)=\top$.
\end{itemize}
 One may easily check that for every $i,j\in\{1,2,3\}$,
$(N_i,s_i)$ is $b$-bisimilar to  $(N_j,s_j)$. Also, for every
$i,j\in\{1,2,3\}$, $(N_i,s_i)$ is $a$-bisimilar to  $(N_j,s_j)$ if
and only if $i=j$.

\noindent Therefore,  the action model associated to
$L_{ab}(\alpha,\beta,\gamma)$ is $(N,1)$, where $S=\{1,2,3\}$,
$1\rightharpoonup_a 1$, $2\rightharpoonup_a 2$,
$3\rightharpoonup_a 3$, and  $1\rightharpoonup_b 1$,
$2\rightharpoonup_b 2$, $3\rightharpoonup_b 3$, and
$1\rightharpoonup_b 2$, $2\rightharpoonup_b 3$,
$3\rightharpoonup_b 1$, and $pre(1)=p$, $pre(2)=\neg p$, and
$pre(3)=\top$.

\item[4-] (\textbf{bothmayread}). Bert orders a drink at the bar
while Anne goes to the bathroom. Both may have read the letter
(and actually both of them have read).
\begin{center}
$L_{ab}(L_a ?p \cap L_b?p, L_a?\neg p\cap L_b?\neg p,L_a ?\top
\cap L_b?\top, L_a?p\cap p|_b L_b?\top, L_a \neg p\cap \neg p|_b
L_b?\top, L_b?p\cap p|_a L_a?\top, L_b \neg p\cap \neg p|_a
L_a?\top)$.
\end{center}

\end{itemize}

Four above information change scenarios are also modelled in
concurrent dynamic epistemic logic in example~7,
of~\cite{kn:dit2}. Also, in Figure~1 of~\cite{kn:dit2}, page~4,
Epistemic states  resulted from the execution of actions (for
these scenarios) described in concurrent dynamic epistemic logic
(example~7, of~\cite{kn:dit2}) is shown. It easy to verify that
\begin{quote}
\emph{Epistemic states  resulted from the execution of pointed
action models associated to basic learning programs for these
scenarios (introduced above) are exactly the same Epistemic states
resulted from the execution actions described in concurrent
dynamic epistemic logic~\cite{kn:dit2} (and shown in Figure~1 of
the same reference, in page~4).}
\end{quote}

\section{Learning by Recursion}\label{LBR}
By \emph{learning by recursion}, we mean the cases where agents
learn about each other's learning, i.e., an agent $a$ learns
something about learning of another agent $b$ and agent $b$ also
learns about learning of agent $a$.  In this way, a recursive
learning occurs. In this section, we introduce recursive learning
actions to model this type of learning.

In the following definition, $undf$, indicates that the
\emph{term} is undefined.

\begin{definition}\label{opterm}
Let $\Phi$ be a set of epistemic formulas over a set of atomic
formulas $P$ and a set of agents $A$, and $Var=\{X,Y,Z,...\}$ be a
set of variables. The set of \emph{open terms}
$\mathrm{OpenT}(\Phi)$ is defined as follows:

\begin{itemize}

\item[\emph{1}.] for all $X\in Var$, $X$ is an open term,  we let
$pre(X)=undf$ and $group(X)=undf$,

\item[\emph{2}.] for all $\varphi\in \Phi$, $?\varphi$ is an open
term, $group(?\varphi)=\emptyset$, and $pre(?\varphi)=\varphi$,

\item[\emph{3}.] for all $n\in \mathbb{N}$, and
$B\subseteq A$, if $\alpha_1,\alpha_2,...,\alpha_n$ are open
terms, then $L_B(\alpha_1,\alpha_2,...,\alpha_n)$ is an open
term, and we let
\begin{center}
$group(L_B(\alpha_1,\alpha_2,...,\alpha_n))=B\cup
group(\alpha_1)$, and

$pre(L_B(\alpha_1,\alpha_2,...,\alpha_n))=pre(\alpha_1)$
\end{center}
Note that the left sides are defined if $group(\alpha_1)$ and
$pre(\alpha_1)$ are defined.
\item[\emph{4}.] $\alpha_1\cap\alpha_2$ is an open term whenever
$\alpha_1$ and $\alpha_2$ are open terms. If both
$group(\alpha_1)$, $group(\alpha_2)$ are defined and
$group(\alpha_1)\cap group(\alpha_2)=\emptyset$, and if both
$pre(\alpha_1),pre(\alpha_2)$ are defined and
$pre(\alpha_1)=pre(\alpha_2)$. Then we let
$group(\alpha_1\cap\alpha_2)=group(\alpha_1)\cup group(\alpha_2)$
and $pre(\alpha_1\cap\alpha_2)=pre(\alpha_1)$,

\item[\emph{5}.] $\psi|_B\alpha_1$ is an open term whenever
$\alpha_1$ is an open term. If $group(\alpha_1)$ is defined and
$B\subseteq group(\alpha_1)$. We then define
$group(\psi|_B\alpha_1)= B$ and $pre(\psi|_B\alpha_1)=\psi$.

\end{itemize}
\end{definition}

\begin{definition}
Assume $\alpha(X_1,X_2,...,X_k)$ is an open term with $k$
variables, the tuple
\begin{center}
$(\beta_1,\beta_2,...,\beta_k)\in \mathrm{OpenT}(\Phi)$
\end{center}
is a suitable substitution for $\alpha(X_1,X_2,...,X_k)$, whenever
$\alpha(\beta_1,\beta_2,...,\beta_k)\in \mathrm{OpenT}(\Phi)$.
\end{definition}

Now we define \emph{open action models}, in order to be associated
to open terms. For each variable $X$, let $(N_X,s_X)$ be a
variable pointed action model, where $N_X=\langle
S^X,(\rightharpoonup^X_a)_{a\in A}, pre^X\rangle$ and the set
$S_X$, the relations $\rightharpoonup^X_a$, and the function
$pre^X$ are \emph{variables}. The \emph{open action model} of an
open term $\alpha(X_1,X_2,...,X_k)$ is simply constructed using
variable pointed action models $(N_{X_{1}},s_{X_{1}})$,
$(N_{X_{2}},s_{X_{2}})$, ..., and $(N_{X_{k}},s_{X_{k}})$ and
Definition~\ref{defsem}. Substituting an action model $(N,s)$ in
an open term $\alpha(X)$ is obtained by considering $s$ instead
of $s_X$.

\begin{example}{\em
Suppose $A=\{a,b\}$, and consider the open term $L_bL_a(X)$. The
open action model $(N_{L_bL_a(X)},s_{L_bL_a(X)})$  is constructed
as follows. We consider three states $s_0,s_1$ and $s_X$, where
$pre(s_0)=pre(s_1)=pre(s_X)$. The actual state is $s_0$, and the
accessibility relations are $s_0\rightharpoonup_b s_0$,
$s_0\rightharpoonup_a s_1$, $s_1\rightharpoonup_a s_1$, and ,
$s_1\rightharpoonup_bt$ for all $t\in b[s_X]$, where $b[s_X]$ is
the set of all states in the hypothetical model $N_X$ in which
$s_X$ has $b$-accessibility to them. The open action model of the
term $L_bL_a(X)$ is  illustrated in Figure~8.
\begin{center}\scalebox{.7}{\includegraphics{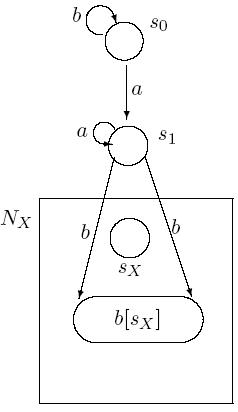}}

Figure~8
\end{center}

\item[1.] Let $(M,t_0)$ be the action model of the term
$L_{ab}(?\varphi)$, and $S=\{t_0\}$, $pre(t_0)=\varphi$,
$t_0\rightharpoonup_a t_0$ and $t_0\rightharpoonup_b t_0$.

\begin{center}\scalebox{.6}{\includegraphics{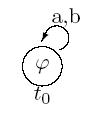}}

Figure~9
\end{center}
Substitution of $(M,t_0)$ for $(N_X,s_X)$ means to substitute
$t_0$ for $s_X$. Therefore, $pre(s_X)=pre(t_0)=\varphi$ and
$b[s_X]=b[t_0]=\{t_0\}$. Substituting $(M,t_0)$ in
$(N_{L_bL_a(X)},s_{L_bL_a(X)})$ yields to the following action
model
\begin{center}\scalebox{.6}{\includegraphics{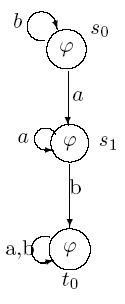}}

Figure~10
\end{center} which is the action model of
$L_aL_bL_{ab}(?\varphi)$.
}
\end{example}

\begin{example}\label{opfx}{\em
Suppose $A=\{a,b\}$, and consider the open term
$L_b(\varphi|_aL_a(\psi|_bL_b(X)))$. The open action model
$(N_{L_b(\varphi|_aL_a(\psi|_bL_b(X)))},s_{L_b(\varphi|_aL_a(\psi|_bL_b(X)))})$
is shown in Figure~11.
\begin{center}\scalebox{.6}{\includegraphics{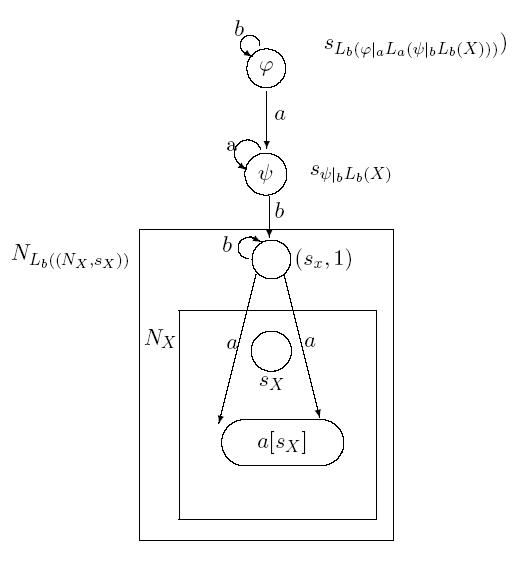}}

Figure~11
\end{center}

$(N_{L_b((N_X,s_X))},s_{L_b((N_X,s_X))})$ is obtained by adding
the new state $(s_X,1)$ as the actual state and adding new
accessibility relations $(s_X,1)\rightharpoonup_b (s_X,1)$, and
$(s_X,1)\rightharpoonup_a t$, for all $t\in a[s_X]$. We have
$pre((s_X,1))=pre(s_X)$, and $s_{L_b((N_X,s_X))}=(s_X,1)$.

$(N_{\psi|_bL_b(X)},s_{\psi|_bL_b(X)})$ is obtained by adding the
new state $s_{\psi|_bL_b(X)}$ to $N_{L_b((N_X,s_X))}$, and adding
new accessibility relations $s_{\psi|_bL_b(X)}\rightharpoonup_b
t$, for all $t\in b[s_{L_b((N_X,s_X))}]$. Moreover,
$pre(s_{\psi|_bL_b(X)})=\psi$.

Continuing the above scenario and using Definition~\ref{defsem},
the open action model in Figure~11 is constructed.

To obtain the open action model of
$L_b(\varphi|_aL_a(\psi|_bL_b(L_b(\varphi|_aL_a(\psi|_bL_b(X))))))$,
one may simply consider \emph{two} copies of the open action model
$(N_{L_b(\varphi|_aL_a(\psi|_bL_b(X)))}$ and replace the actual
state $s_{L_b(\varphi|_aL_a(\psi|_bL_b(X)))}$ of one of them for
$s_X$ in another one, and obtain the model shown in Figure~12. By
substituting $s_{L_b(\varphi|_aL_a(\psi|_bL_b(X)))}$ for $s_X$,
the set $a[s_x]$ is $\{s_{\psi|_bL_b(X)}\}$.

\begin{center}\scalebox{.6}{\includegraphics{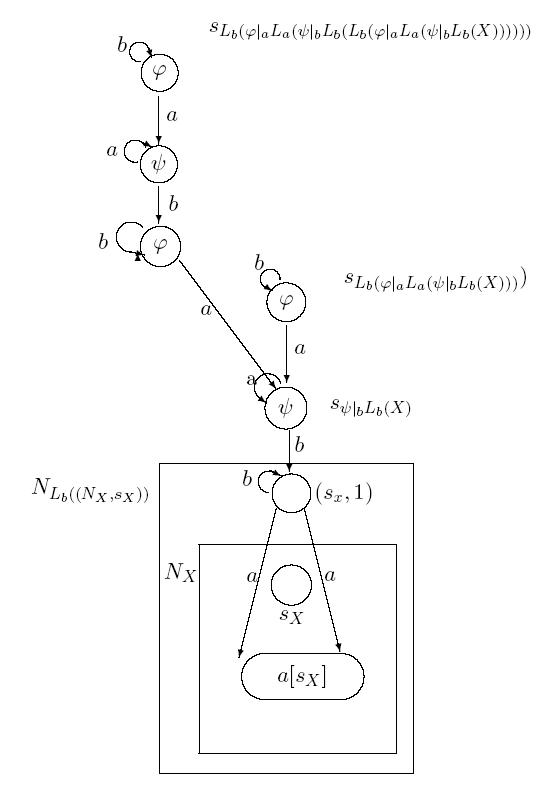}}

Figure~12
\end{center}

\textbf{Fixed point of} $L_b(\varphi|_aL_a(\psi|_bL_b(X)))$. To
construct an action model as a \emph{fixed point} of the term
$L_b(\varphi|_aL_a(\psi|_bL_b(X)))$, we consider the open action
model $(N_{L_b(\varphi|_aL_a(\psi|_bL_b(X)))}$ and replace its
actual state $s_{L_b(\varphi|_aL_a(\psi|_bL_b(X)))}$ for $s_X$
(see Figure~13, the symbol $s_{\mu} X$ will be defined in the next
section).
\begin{center}
\scalebox{.6}{\includegraphics{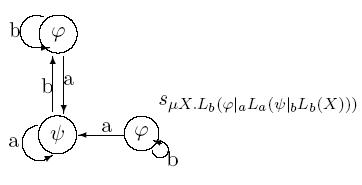}}

Figure~13
\end{center}

It is easy to verify that the pointed action model in Figure~13 is
bisimilar to the pointed action model $(N_3,s)$ explained in
Example~\ref{rexm} as a recursive learning action.}
\end{example}

\subsection{Recursive Learning Programs}
One may check that the graph of the action model shown in Figure~3
is not a tree and thus by Theroem~\ref{prtree1}, it is not a basic
learning action. So we cannot describe it in terms of alternative
learning, $L_B(-,-,...,-)$, concurrent learning $\cap$, and wrong
learning, $|_B$, operators.  We add a new operator $\mu$ to the
language for recursive learning and show that the action model
shown in Figure~3 is a recursive learning action. To do this, we
need to slightly modify the Definition of open terms~\ref{opterm}
as follows.
\begin{definition}
Let $\Phi$ be a set of epistemic formulas over a set of atomic
formulas $P$ and a set of agents $A$, and $Var=\{X,Y,Z,...\}$ be a
set of variables. The set of \emph{ open terms},
$\mathrm{OT}(\Phi)$ is defined as follows:

\begin{itemize}

\item[\emph{1}.] for all $X\in Var$, $X$ is an  open term, we let
$pre(X)=undf$ and $group(X)=undf$,

\item[\emph{2}.] for all $\varphi\in L(\Phi)$, $?\varphi$ is an
open term, $group(?\varphi)=\emptyset$, and
$pre(?\varphi)=\varphi$,

\item[\emph{3}.] for all $n\in \mathbb{N}$, and
$B\subseteq A$, if $\alpha_1,\alpha_2,...,\alpha_n$ are open terms
then $L_B(\alpha_1,\alpha_2,...,\alpha_n)$ is an open term, and
we let
\begin{center}
$group(L_B(\alpha_1,\alpha_2,...,\alpha_n))=B\cup
group(\alpha_1)$ and,

$pre(L_B(\alpha_1,\alpha_2,...,\alpha_n))=pre(\alpha_1)$
\end{center}
Note that the left sides are defined if $group(\alpha_1)$ and
$pre(\alpha_1)$ are defined.
\item[\emph{4}.] $\alpha_1\cap\alpha_2$ is an open term whenever
$\alpha_1$ and $\alpha_2$ are open terms. If both
$group(\alpha_1)$, $group(\alpha_2)$ are defined and
$group(\alpha_1)\cap group(\alpha_2)=\emptyset$, and if both
$pre(\alpha_1),pre(\alpha_2)$ are defined and
$pre(\alpha_1)=pre(\alpha_2)$. Then we let
$group(\alpha_1\cap\alpha_2)=group(\alpha_1)\cup group(\alpha_2)$
, and $pre(\alpha_1\cap\alpha_2)=pre(\alpha_1)$,

\item[\emph{5}.] $\psi|_B\alpha_1$ is an open term whenever
$\alpha_1$ is an open term. If $group(\alpha_1)$ is defined and
$B\subseteq group(\alpha_1)$, We define $group(\psi|_B\alpha_1)=
B$ and $pre(\psi|_B\alpha_1)=\psi$.

\item[\emph{6}.] If $\alpha(X_1,X_2,...,X_k)$ is an open term,
and both
\begin{center}
$group(\alpha(X_1,X_2,...,X_k))$ and\\
$pre(\alpha(X_1,X_2,...,X_k))$
\end{center}
are defined, then $\mu X_1.\alpha(X_1,X_2,...,X_k)$ is an open
term with $k-1$ free variables; the variable $X_1$ is bound under
$\mu X_1$. We define
\begin{center}
$group(\mu
X.\alpha(X_1,X_2,...,X_k))=group(\alpha(X_1,X_2,...,X_k))$ and\\
$pre(\mu X.\alpha(X_1,X_2,...,X_k))=pre(\alpha(X_1,X_2,...,X_k))$.
\end{center}
\end{itemize}

A term $\alpha\in\mathrm{OT}(\Phi)$ is \emph{closed} if it has no
unbounded variable.
\end{definition}

\begin{definition}
Let $\Phi$ be a set of epistemic formulas.  The set of \emph{
recursive learning programs}, $\mathrm{RLP}(\Phi)$, is the set of
all closed terms in $\mathrm{OT}(\Phi)$.
\end{definition}

\begin{definition}\label{fxp}({\bf Semantics of $\mathbf{\mu
X.\alpha(X})$}). The associated pointed action model to  $\mu
X.\alpha(X)$ is obtained by replacing the actual state
$s_{\alpha(X)}$ of the action model $N_{\alpha(X)}$ for the state
$s_X$.
\end{definition}

The associated action model of $\mu X.\alpha(X)$ is actually a
{\em fixed point} of the open action model of $\alpha(X)$, see
Example~\ref{opfx}.

\begin{example}{\em The semantics of the  recursive
learning program
\begin{center} $\alpha:=\mu X_1. L_c(\chi|_a\mu
X_2.L_a(\varphi|_bL_b((\psi|_aX_2)\cap(\psi|_cX_1))))$
\end{center}
is constructed in the following way. The model of the
open term
\begin{center}$L_a(\varphi|_bL_b((\psi|_aX_2)\cap(\psi|_cX_1)))$
\end{center} is shown
in Figure~14.
\begin{center}
\scalebox{.6}{\includegraphics{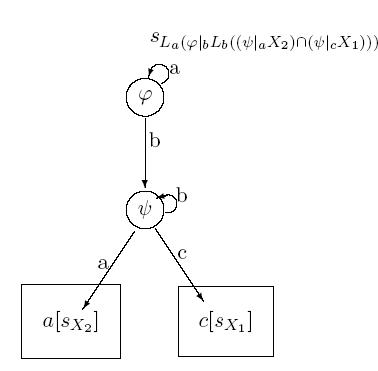}}

Figure~14
\end{center}
The model of the open term
\begin{center}
$\mu X_2.L_a(\varphi|_bL_b((\psi|_aX_2)\cap(\psi|_cX_1)))$
\end{center}
is shown in Figure~15.
\begin{center}
\scalebox{.6}{\includegraphics{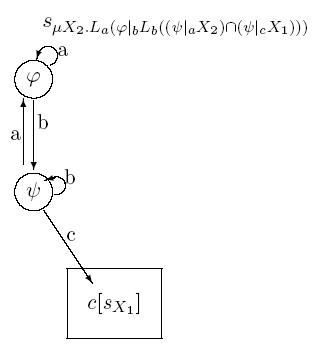}}

Figure~15
\end{center}
The action model of
\begin{center}
$L_c(\chi|_a\mu
X_2.L_a(\varphi|_bL_b((\psi|_aX_2)\cap(\psi|_cX_1))))$
\end{center}
is illustrated in Figure~16.
\begin{center}
\scalebox{.6}{\includegraphics{fig16.jpg}}

Figure~16
\end{center}
Finally the action model of the recursive learning program
$\alpha$ is shown in Figure~17.
\begin{center}
\scalebox{.6}{\includegraphics{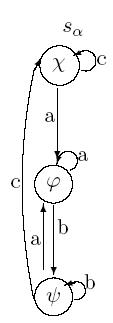}}

Figure~17
\end{center}
 }
\end{example}

As two other examples, one may check that the action model
associated to the recursive learning program
\begin{center}$\mu X.L_a(\varphi|_b\mu
Y. L_b(\psi|_aX\cap\psi|_cL_c(\theta|_bY)))$
\end{center}
is the pointed action model illustrated in Figure~18,

\begin{center}
\scalebox{.6}{\includegraphics{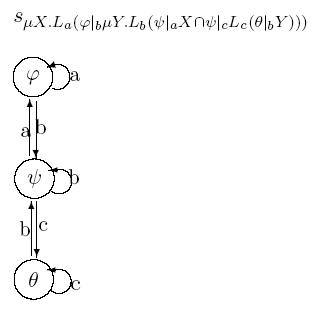}}

Figure~18
\end{center}
and the action model associated to the recursive learning program
\begin{center}
$\mu X.
L_a(\varphi|_bL_b(\varphi|_aX),\psi|_bL_b(\varphi|_aX))$
\end{center}
is bisimilar to the pointed action model $(N,s)$ illustrated in
Figure~19.

\begin{center}
\scalebox{.6}{\includegraphics{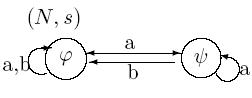}}

Figure~19
\end{center}

\subsection{Two Main Theorems}
In this subsection, we present two main theorems of the paper. In
the first one, Theorem \ref{main}, we show that every finite $K45$
action model is associated to some recursive learning program and
conversely, for every learning program, there is a finite $K45$
action model associated with it. Then we introduce a hierarchy
over learning programs with respect to the number of recursive
learning operators. In our second main Theorem \ref{main2}, it is
shown that the hierarchy of the learning programs is strict, i.e.,
it is not possible to describe all $K45$ action models by a
determined finite number of recursion in learning. That means that
the hierarchy does not collapse.

\subsubsection{Representing Epistemic Action Models}\label{rep}

\begin{definition}\label{closeds}
Let $N=\langle S,(\rightharpoonup_a)_{a\in A}, pre\rangle$ be a
$K45$ action model. For each agent $a\in A$, an $a$-component of
$N$ is $M_a=\langle S',(\rightharpoonup'_a)_{a\in \{a\}},
pre'\rangle$, where $M_a$ is an $S5$ closed connected submodel of
$N$.
\begin{itemize}
\item[]  Assume for all $a\in A$, $M^i_a=\langle
S^i_a,(\rightharpoonup^i_a)_{a\in \{a\}}, pre^i_a\rangle$, $1\leq
i\leq k_a$ are all different $a$-component of $N$. We construct
the action model $T'(N)=\langle S',(\rightharpoonup'_a)_{a\in A},
pre'\rangle$ as follows:
\begin{itemize}
\item $S'=\bigcup_{a\in A} \bigcup_{1\leq i\leq
k_a}(S^i_a\times\{(i,a)\})$,

\item $(s,(i,a))\rightharpoonup'_b (t,(j,b))$ if and only if
either $i=j$, $a=b$ and $s\rightharpoonup^i_a t$, or $i\neq j$,
$b\neq a$, and $s\rightharpoonup_b t$ (in the action model $N$),

\item $pre'((s,(i,a)))=pre^i_a(s)$.
\end{itemize}
\item[] The \emph{projection} of an state $(s,(i,a))$ in the
action model $T'(N)$, denoted by $\Pi((s,(i,a)))$, is defined to
be the state $s$ in the action model $N$.

\end{itemize}
\end{definition}

\begin{example}{\em
Consider the action models $N_1$ and $N_2$ in Figure~20. Their
action models $T'(N_1)$ and $T'(N_2)$ are illustrated in Figure~20
as well.
\begin{center}
\scalebox{.6}{\includegraphics{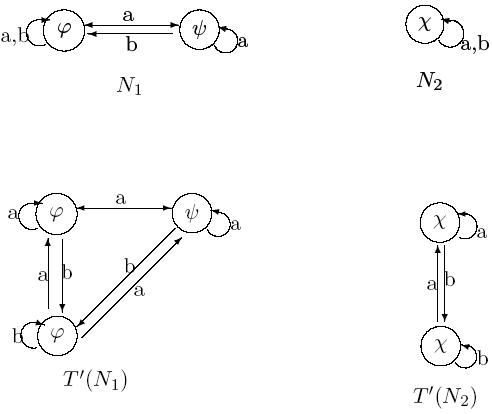}}

Figure~20
\end{center}
}
\end{example}

\begin{proposition}
Let $N=\langle S,(\rightharpoonup_a)_{a\in A}, pre\rangle$ be an
$K45$ action model. Then $(N,s_0)$ is bisimilar to $(T'(N),w_0)$
for all $s_0$ and $w_0$, in which the state $s_0$ is the
projection of $w_0$.
\end{proposition}
\begin{proof} The proof is similar to the proof of
Proposition~\ref{ddd}.
\end{proof}

\begin{theorem}\label{main}
All finite epistemic actions are recursive learning programs,
i.e.,
\begin{center}
$\mathrm{FAct}(\Phi)=\mathrm{RLP}(\Phi)$.
\end{center}
\end{theorem} \begin{proof} See the Appendix.\end{proof}

\begin{example}{\em It is shown in Figure~21, how to construct a
program for the pointed action model $(N_1,s)$ through the
instruction argued in proof of the Theorem~\ref{main}.

\begin{center}
\scalebox{.6}{\includegraphics{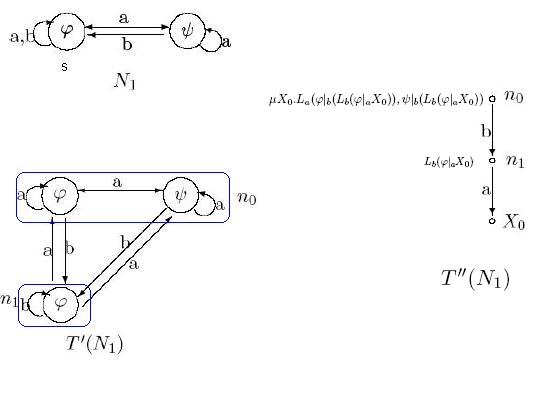}}

Figure~21
\end{center}
}
\end{example}

\subsubsection{A Hierarchy of Learning}

\begin{definition}
For each $k\in\mathbb{N}$, we define the class $\mathrm{kRLP}$, of
all finite $K45$ pointed action models which can be described by a
recursive learning program with at most $k$ times of
\emph{dependent} use of the recursive operator $\mu$.
\end{definition}
The term `\emph{dependent}' in the above definition is crucial. In
a program, two operators $\mu X$ and $\mu Y$ are called to be
dependent if it is not possible to use one variable for both
operators and achieve the same action model.

For example, in the program,
\begin{center}
$(\chi|_a \mu X. L_a(\varphi|_b L_b(\psi|_aX)))\cap
(\chi|_b\mu Y. L_b(\psi|_aL_a(\varphi|_bY))$,
\end{center}
the operators $\mu X$ and $\mu Y$ are independent. Note that the
program
\begin{center}$(\chi|_a \mu Z. L_a(\varphi|_b L_b(\psi|_aZ)))\cap
(\chi|_b\mu Z. L_b(\psi|_aL_a(\varphi|_bZ))$,
\end{center}
describes the same action model. In contrast with the above
example, in the following program,
\begin{center}
$\mu X.L_a(\varphi|_b\mu Y.
L_b(\psi|_aX\cap\psi|_cL_c(\theta|_bY)))$,
\end{center}
the operators $\mu X$ and $\mu Y$ are dependent.

It is easy to observe that the class $\mathrm{0RLP}$ is the class
of all basic learning programs, $\mathrm{BLP}$. We also wish to
name $\mathrm{1RLP}$ as the class of primitive recursive learning
programs and denote it also by $\mathrm{PRLP}$.

To prove the next theorem, we need to clarify some notions in
graph theory. Let $G=(N,E)$ be a directed graph. A {\em simple
loop} $L$ in the graph $G$ is a sequence of nodes $\langle s_0,
s_1, s_2,...,s_n\rangle$, such that for all $i<n$,
$(s_i,s_{i+1})\in E$, $(s_n,s_0)\in E$, and for all $0\leq
i,j\leq n$, if $i\neq j$ then $s_i\neq s_j$. We call $s_0$ the
\emph{start-point} of the simple loop $L$. Let $L=\langle s_0,
s_1, s_2,...,s_n\rangle$ and $L'=\langle s'_0, s'_1,
s'_2,...,s'_m\rangle$ be two simple loops. We say $L'$ is {\em
connected} to $L$ by its start-point, if there exists $0\leq t\leq
n$, $s_t\neq s_0$ and $s_t=s'_0$, and for all $0\leq i\leq n$ and
$0\leq j\leq m$ if $i\neq t$ and $j\neq 0$, then $s_i\neq s_j$. A
$k$-{\em nested loop} is a sequence $\langle
L_1,L_2,...,L_k\rangle$ of simple loops, such that for each
$i>1$, $L_i$ is connected to $L_{i-1}$ by its start-point.

\begin{theorem}\label{main2}
We have the following hierarchy of recursive learning programs
\begin{center}
$\mathrm{BLP}\varsubsetneq \mathrm{PRLP}\varsubsetneq
\mathrm{2RLP}\varsubsetneq ...\varsubsetneq
\mathrm{kRLP}\varsubsetneq \mathrm{(k+1)RLP}\varsubsetneq...$.
\end{center}
Moreover, it does not collapse to $\mathrm{RLP}$, i.e, for all
$k$,
\begin{center}
$\mathrm{kRLP}\neq \mathrm{RLP}$,
\end{center}
\end{theorem}\begin{proof} See the Appendix.\end{proof}

\section{Concluding Remarks and Further Work}\label{CRFW}

\subsection{Related Works}\label{scon}

We may compare \emph{epistemic learning programs} with other
approaches, like \emph{concurrent dynamic epistemic
logic}~\cite{kn:dit2} and \emph{epistemic programs}~\cite{kn:bal}.

In \emph{concurrent dynamic epistemic logic}~\cite{kn:dit2}, an
epistemic action is interpreted as a \emph{relation} between $S5$
epistemic states and sets of $S5$ epistemic states. There are two
main differences between the interpretation of epistemic action in
concurrent dynamic epistemic logic  and epistemic learning
programs.

\begin{itemize}

\item[1.] An epistemic action in concurrent dynamic epistemic
logic is a \emph{relation} between epistemic states whereas in
\emph{epistemic learning programs}, it is a \emph{function} from
epistemic states to epistemic states.

\item[2.] Concurrent dynamic epistemic logic is just about $S5$
models whereas epistemic learning programs also considers $K45$
models.

\end{itemize}

Another difference is in the interpretation of the notion of
learning. Our learning operator is an operator on action models,
and  $L_B(\alpha_1,\alpha_2,...,\alpha_k)$ is a new action model,
expressing the condition that agents in $B$ learn that an action
among $\alpha_1,\alpha_2,...,\alpha_k$ has occurred, whereas the
action $\alpha_1$ has \emph{actually} occurred. For example,
$L_b(?\varphi,?\psi)$ is an action model which says:
\begin{quote}
\emph{$\varphi$ is announced and agent $b$ is suspicious whether
he learns $\varphi$ or learns $\psi$.}
\end{quote}
One may compare the above learning program with the action
$L_b(!?\varphi\cup?\psi)$ in dynamic epistemic logic, and observe
that for all $S5$ epistemic state $(M,s)$, we have that
$(M,s)\times(N_{L_b(?\varphi,?\psi)},t_{L_b(?\varphi,?\psi)})$ is
bisimilar to $(M,s)[[L_b(!?\varphi\cup?\psi)]]$. However, note
that the learning program $L_b(L_b(?\varphi,?\psi))$ is equal to
the learning program $L_b(?\varphi)$, whereas, in dynamic
epistemic logic, the action $L_b(L_b(!?\varphi\cup?\psi))$ is
equal to $L_b(!?\varphi\cup?\psi)$. Thus there is a difference
between the notion of learning we consider for learning programs
and the notion of learning considered in dynamic epistemic logic.

Despite of the above arguments, it seems possible to translate a
class of action terms, say $\alpha$, in concurrent dynamic
epistemic logic to a recursive learning program
$\mathrm{tr}(\alpha)$, such that for all $S5$ epistemic state
$(M,s)$, we have that
$(M,s)\times(N_{\mathrm{tr}(\alpha)},t_{\mathrm{tr}(\alpha)})$ is
bisimilar to $(M,s)[[\alpha]]$.

Another way to represent information change is via the notion of
epistemic program introduced in~\cite{kn:bal}. Whereby the notion
of \emph{action signature} is introduced and by adding this
notion to the propositional dynamic logic
$PDL$~\cite{kn:pdl2,kn:pdl}, a logical language is obtained to
represent information change. However in this setting, no learning
operator is considered, and the information change is represented
through \emph{action signature}, \emph{alternative},
\emph{sequential}, and \emph{iteration} compositions. We focus on
different kinds of learning; as the primitive notion of
information change is learning something by agents.

Another  work related to ours is~\cite{kn:pen} where the epistemic
programs are discussed by adding a parallel composition operator
to non-deterministic sum and sequential composition.

We showed that all finite $K45$ action models can be described by
recursive learning programs. It is also announced in~\cite{kn:BV}
that every $S5$ action model can be described as a concurrent
epistemic action.

 We
showed that $K45$ models are models of the belief $KD45$ logic for
applicable formulas. In this way, to preserve the belief
consistency of an agent, the agent is absent in the states that
conflicts his beliefs. A similar work has been done
in~\cite{kn:david}, which assumes that a rational agent rejects
those incoming information which dispute his beliefs.

By introducing   $K45$ models and actions, we may think of a
theory of multi-agent belief revision.  A related work
is~\cite{kn:guil}, which generalize AGM~\cite{kn:agm}, to a
multi-agent belief revision theory.

Our work presents a method to construct $K45$ action models
through some basic constructors. Also in~\cite{kn:wang}, it is
introduced  operators  to compose epistemic models in order to
construct large models by small components representing agents'
partial observational information.
\subsection{Further work}

\subsubsection{A functional Semantics}

As a semantics of epistemic learning programs, we associated a
pointed action model to every basic learning program. We may
propose a \emph{functional semantics} for the basic learning
programs, in the manner that each program is associated to a
\emph{partial} function from epistemic states $Mod$ to $Mod$. In
this semantics, the meaning of learning operator is
\emph{different} form the meaning we discussed in the
introduction. Here, learning in epistemic states $(M,s)$ deals
with two things, a set $U$ of states of $M$. which includes the
actual state $s$, \emph{and} a set of agents $B\subseteq A$ (where
$A$ refers to the set of all agents). Learning with $B\subseteq A$
and $U\subseteq S$ in the epistemic states $(M,s)$ means that:

\begin{quote}
\emph{agents in $B$ become  aware that the actual state is among
the states in $U$, and other agents in $A-B$ believe that nothing
has occurred.}
\end{quote}
Let $\Phi$ be a set of epistemic formulas over a set of atomic
formulas $P$ and a set of agents $A$. To each $\alpha\in
BLP(\Phi)$, we associate a pair $(f_\alpha, U_\alpha)$, where
$f_\alpha: Mod\rightarrow Mod$ is a partial function ,and for each
epistemic state $(M,s)$, $U_\alpha((M,s))$ is a subset of $S'$,
where $f_\alpha((M,s))=(M',s')$ with $M'=\langle S',
(\rightharpoonup'_i)_{i\in A},V'\rangle$.

For a recursive learning program $\mu X.\alpha(X)$, the
associated partial function should satisfy the fixed point
equation, i.e., $f_{\mu X.\alpha(X)}=f_{\alpha}\circ f_{\mu
X.\alpha(X)}$. As our forthcoming work, we aim to study this
functional semantics. It seems to us that  functional semantics
and recursive learning take us beyond the action models, that is,
by functional semantics, epistemic learning programs can encodes
information changes which cannot be encode by action models.

\subsubsection{A Logic for RLP}
We need to provide a proof system for $\mathrm{RLP}$  as it is
done for other approaches, like concurrent dynamic epistemic
logic ~\cite{kn:dit2}, and action models~\cite{kn:bal1,kn:bal}.

\subsubsection{Notions of Learning}
In Introduction, we put forward two meanings for   1.
\emph{pointed action models} (see~\ref{mean1}) and 2.
\emph{learning of an action model} (see~\ref{mean2}). We supposed
that an action model describes what is announced  and what agents
perceive based on their accesses to the resource of announcement.
We also assumed that the learning of an action by a set of agents
is to learn about the way  information change. So our meanings of
\emph{action models} and \emph{learning} refer  to the
\emph{occurrence} of information change.

We may propose two other meanings for pointed action models and
learning of an action model, which refer  to \emph{disability} in
information change. In this way, an action model describes the
disability of agents in hearing or accessing  the resource of
announcement. For example, the new meaning of the pointed action
model $(N_1,s)$ in Figure~1,
\begin{center}\scalebox{.6}{\includegraphics{fig1.jpg}}

Figure~1
\end{center}
is
\begin{quote}
``in the case of announcement of $\varphi$   agent $a$ hears
$\psi$. "  \end{quote} Note that the above meaning does not speak
about what occurs in information change, but it just describes a
disability of   agent $a$. Suppose $\varphi=green$ and
$\psi=blue$.  The new meaning of the pointed action model
$(N_1,s)$  is that   agent $a$ has a color-blindness and if   a
green ball is shown to her then she thinks that she sees a blue
ball. Similarly, the meaning of learning an action changes. The
new meaning is learning about \emph{disability} not about
\emph{occurrence}. The learning of  an action by a set of agents
is to learn about the disabilities  that the agents have.
Figure~22 shows two pointed action models where both refer to
$L_a((N_1,s))$ (agent $a$ learns the pointed action model
$(N_1,s)$), but one considers the \emph{occurrence meaning} and
the other considers the \emph{disability meaning}.

\begin{center}
\scalebox{.6}{\includegraphics{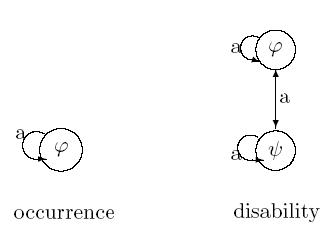}}

Figure~22
\end{center}
In the \emph{occurrence meaning},  agent $a$ learns that $\varphi$
is announced. In the \emph{disability meaning},   agent $a$
becomes aware of her color-blindness, and after this learning, if
she sees
  a blue ball, she is suspicious whether it is green or blue.

\vspace{.2in}

\noindent  \textbf{Acknowledgement}. The authors would like to
thank Hans van Ditmarsch for his careful reading of our
manuscript, and his very helpful comments and suggestions. We
also would like to thank Mehrnoosh Sadrzadeh for her very helpful
comments and suggestions to improve the readability of the
manuscript.

\section{Appendix}
\begin{lemma}\label{lem23} Let $M=\langle S,(\rightharpoonup_a)_{a\in A}, V\rangle$ be a
Kripke model such that for each  $a\in A$, $\rightharpoonup_a$ is
\emph{Euclidean}. Then   for
all $s\in S$, if there is a state $v\in S$ such that
$v\rightharpoonup_a s$ then there exists $t\in S$ such that
$s\rightharpoonup_a t$.
\end{lemma}
\begin{proof}
Let $S$ be a set of states and $\rightharpoonup \subseteq S\times
S$ be an \emph{Euclidean} relation. Suppose $s, v\in S$ are
arbitrary and $v\rightharpoonup s$. By \emph{Euclidean} property,
we derive $s\rightharpoonup s$, and we are done.
\end{proof}

\noindent \begin{proof}\textbf{\ref{ddd}.} Consider $T(N)=\langle
S',(\rightharpoonup'_a)_{a\in A}, pre'\rangle$. Define
$\mathcal{R}\subseteq S\times S'$ as follows. For all $s\in S$ and
$w\in S'$,
\begin{center}
$s\mathcal{R}w $ if and only if $\Pi(w)=s$.
\end{center}
We show that $\mathcal{R}$ is a bisimulation relation. Suppose
all different the maximal closed connected $S5$ submodels of $N$
are $M^1,M^2,...$ and $M^k$. Assume $s\mathcal{R}w $. Then
$w=(s,i)$, for some $1\leq i\leq k$.

\begin{itemize}
\item[-] Forth. Let $s\rightharpoonup_a t$. Either $a\in B^i$ or
$a\not \in B^i$. In the first case, since $M^i$ is closed, we have
$t\in S^i$ and thus $(s,i)\rightharpoonup_a' (t,i)$, and since $t
\mathcal{R} (t,i)$, we are done. In the second case, since $N$ is
a $K45$ model, and $s\rightharpoonup_a t$, we have
$t\rightharpoonup_a t$, by Lemma~\ref{lem23}. Therefore there
exists a maximal closed connected $S5$ submodel of $N$, say $M^j$,
such that $t\in S^j$ and $a\in B^j$. By Definition~\ref{closef},
$(s,i)\rightharpoonup_a' (t,j)$, and since $t \mathcal{R} (t,j)$
we are done.

\item[-] Back. Suppose $(s,i)\rightharpoonup_a' (t,j)$. Then by
Definition~\ref{closef}, $s\rightharpoonup_a t$, and we are done.

\item[-] Pre. It is straightforward.
\end{itemize}
\end{proof}

\noindent \begin{proof}\textbf{\ref{prtree1}.}
The proof is by induction on the structure of basic learning
programs.

First of all, note that for each epistemic formula $\varphi$, the
graph of the  action model $(N_{?\varphi},s_{?\varphi})$ is a
tree.

Let $\alpha$ be a basic learning program and its graph
$G(N_\alpha)$ be a tree. We show that the graph
$G(N_{\psi|_B\alpha})$ is a tree, for any arbitrary formula $\psi$
and $B\subseteq group(\alpha)$. The maximal closed connected $S5$
submodels of the action model $N_{\psi|_B\alpha}$ are all the
maximal closed connected $S5$ submodels of $N_{\alpha}$, say
$M^1,M^2,...,M^k$, and the maximal closed connected $S5$ submodel
containing the state $s_{\psi|_B\alpha}$, which is
$M^0=\langle\{s_{\psi|_B\alpha}\},(\rightharpoonup^0_a)_{a\in
\emptyset}, pre^0(s_{\psi|_B\alpha})=\psi\rangle$. One may check
that
\begin{itemize}
\item[1.] for all $i,j\geq 1$, the edges between two nodes $M^i, M^j$
in the graph $G(N_{\psi|_B\alpha})$ are the same edges in the
graph $G(N_\alpha)$,

\item[2.] for all $i\geq 1$, there is no directed edge from $M^i$
to $M^0$, as group of $M^0$ is empty,

\item[3.] for all $i\geq 1$, $M^0\rightarrowtail M^i$ if and only
if there exists $t\in S^i$, $a\in B^i\cap B$, such that
$s_\alpha\rightharpoonup_a t$ in the model $N_\alpha$.
\end{itemize}
Hence, if the graph $G(N_\alpha)$ has no loop, the graph of
$G(N_{\psi|_B\alpha})$ would have no loop as well.

Let $\alpha_1$ and $\alpha_2$ be two basic learning programs such
that $group(\alpha_1)\cap group(\alpha_2)=\emptyset$, and
$pre(\alpha_1)=pre(\alpha_2)$ and $G(N_{\alpha_1})$ and
$G(N_{\alpha_2})$ are trees. Then the graph
$G(N_{\alpha_1\cap\alpha_2})$ is a tree for the following reasons.
The maximal closed connected $S5$ submodels of
$N_{\alpha_1\cap\alpha_2}$ consist of all the maximal closed
connected $S5$ submodels of $N_{\alpha_1}$, and all the maximal
closed connected $S5$ submodels of $N_{\alpha_1}$, and the
maximal closed connected $S5$ submodel containing the state
$s_{\alpha_1\cap\alpha_2}$, which is
$M^0=\langle\{s_{\alpha_1\cap\alpha_2}\},(\rightharpoonup^0_a)_{a\in
\emptyset}, pre^0(s_{\alpha_1\cap\alpha_2})=pre(\alpha_1)\rangle$.
One may check that, since $group(\alpha_1)\cap
group(\alpha_2)=\emptyset$, there is no edge between the nodes of
the subtrees $G(N_{\alpha_1})$ and $G(N_{\alpha_2})$. So if the
graphs $G(N_{\alpha_1})$ and $G(N_{\alpha_2})$ have no loop, the
graph of $G(N_{\alpha_1\cap\alpha_2})$ would have no loop as well.

Assume $\alpha_1$, $\alpha_2$, ... and $\alpha_m$ are basic
learning programs such that their graphs have no loop. The the
maximal closed connected $S5$ submodels of
$N_{L_B(\alpha_1,\alpha_2,...,\alpha_m)}$ are the followings:

\begin{itemize}
\item[1.] the $S5$ model $M_0$ consists of $m$ states
$(s_{\alpha_1},1), (s_{\alpha_2},1),...,(s_{\alpha_m},1)$, with
accessibility relations produced by agent-bisimilarity of group
$B$~(see Definition~\ref{defsem}).

\item[2.] all the maximal closed connected $S5$ submodels of
$N_{\alpha_1}$, $N_{\alpha_2}$, ..., $N_{\alpha_m}$.

\end{itemize}
The node $M_0$ is the root of the graph
$G(N_{L_B(\alpha_1,\alpha_2,...,\alpha_m)})$, and all the graphs
$G(N_{\alpha_1})$, $G(N_{\alpha_2})$, ..., $G(N_{\alpha_m})$ are
disjoint subgraphs of
$G(N_{L_B(\alpha_1,\alpha_2,...,\alpha_m)})$, such that the root
may be connected to them. Hence if $G(N_{\alpha_1})$,
$G(N_{\alpha_2})$, ..., $G(N_{\alpha_m})$ are trees, then
$G(N_{L_B(\alpha_1,\alpha_2,...,\alpha_m)})$ is a tree.
\end{proof}

\noindent \begin{proof}\textbf{\ref{s5prop}.}
Suppose $(N,s_0)$ is an $S5$ pointed action model. Let $N=\langle
S,(\rightharpoonup_a)_{a\in A}, pre\rangle$, where
$S=\{s_0,s_1,...,s_k\}$, and $A=\{a_1,a_2,...,a_m\}$. As the
model is $S5$, all accessibility relations are equivalence
relations. For each agent $a_i$, let $P_i=\{D^1_i,
D^2_i,...,D^{n_i}_i\}$ be the equivalence classes of the relation
$\rightharpoonup_{a_i}$, which partitions the set $S$. Consider
$n_i$ epistemic formulas $\psi_{i,1},\psi_{i,2},...,\psi_{i,n_i}$,
where none of them are $KD45$ equivalent to each other. For each
$s_j$, $0\leq j\leq k$, consider the basic learning program
$\alpha_j$

\begin{center}
$\alpha_j=\beta_{j,1}\cap\beta_{j,2},...,\cap\beta_{j,m}$,
\end{center}
where each $\beta_{j,l}$ is a basic learning program defined as
follows
\begin{center}
$\beta_{j,l}=pre(s_j)|_{a_l}L_{a_l}(?\psi_{l,h})$,
\end{center}
where $s_j\in D^h_l$.

The action model associated to the basic learning program
$L_A(\alpha_0,\alpha_1,\alpha_2,...,\alpha_k)$ is $(N,s_0)$. Since
for each agent $a$, the action models of two programs $\alpha_j$
and $\alpha_i$ are $a$-bisimilar if and only if $s_j$ and $s_i$
are in the same equivalence classes induced  by the relation
$\rightharpoonup_a$.
\end{proof}

\noindent \begin{proof}\textbf{\ref{prtree2}.} By induction on the height of the tree. Let $h_N$ be
the height of the tree $G(N)$. If $h_N=1$, the action model $N$ is
an $S5$ model and by the Proposition~\ref{s5prop}, $(N,s)$ is a
basic learning action. Suppose that for all $h< k$, if the graph
of a $K45$ pointed model $(M,t)$ is a tree with height $h$, then
$(M,t)$ is a basic learning action. Assume $h_N=k$. Let
$M^0,M^1,...,M^{k-1}$ are all maximal connected closed submodels
of $N$, and $s$ is an state of $M^0$. We consider $M^0$ as the
root of the tree, and delete all nodes which are not reachable
from $M^0$. That is because we want to state a program for the
\emph{pointed} action model $(N,s)$, and by deleting those maximal
connected closed submodels which are not reachable from $s$, we do
not loose anything up to bisimilarity. Let $M^0=\langle
S^0,(\rightharpoonup^0_a)_{a\in B^0}, pre^0\rangle$, where
$S^0=\{s^0_1,s^0_2,...,s^0_n\}$, $B^0=\{a^0_1,a^0_2,...,a^0_m\}$
and $s^0_1=s$.  As $M^0$ is an $S5$ action model, by using
Proposition~\ref{s5prop}, there is a basic learning program
$L_{B^0}(\alpha_1,\alpha_2,...,\alpha_{n})$, such that its action
model is $(M^0,s^0_1)$, and each $\alpha_i$ corresponds to a state
$s^0_i$~(see the proof of Proposition~\ref{s5prop}). For each
$s^0_i$, let $\{b^i_1,b^i_2,...,
b^i_{l_i}\}=group_{T(N)}(s^0_i)-B^0$, where $group_{T(N)}(s^0_i)$
is the group of agents of the state $s^0_i$ in the model $T(N)$.
For each $1\leq j\leq l_i$, if $s^0_i\rightharpoonup_{b^i_j}
t_{i,j}$, for some state $t_{i,j}$ of the model $T(N)$, then as
$M_0$ is the root and is not accessible from $t_j$, the graph of
the pointed action model $(T(N),t_{i,j})$ is a tree with height
less than $k$, and by the induction hypothesis, there is a basic
leaning program $\beta_{i,j}$ such that its action model is
$(T(N),t_j)$. For each $s^0_i$, let

\begin{center}
$\gamma_i=\alpha_i\cap (pre(s^0_i)|_{b^i_1}\beta_{i,1})\cap
(pre(s^0_i)|_{b^i_2}\beta_{i,2})\cap...\cap
(pre(s^0_i)|_{b^i_{l_i}}\beta_{i,l_i})$.
\end{center}
The program $L_{B^0}(\gamma_1,\gamma_2,...,\gamma_{n})$ is the
desired program, that is, its associated action model is $(N,s)$.
\end{proof}

\noindent \begin{proof}\textbf{\ref{main}.}
Let $(N,s_0)$ be a $K45$ pointed action model. Consider the
action model $T'(N)$. One may check that the followings hold true.
Suppose $M^a_1$ and $M^b_2$ are two \emph{different} components in
$T'(N)$, then
\begin{itemize}
\item[1.] if $a=b$, as both components are closed $S5$ models,
there is no accessibility relation for agent $a$ between the two
components,

\item[2.] if for some state $s\in M^a_1$ and $t\in M^b_2$, we have
$s\rightharpoonup_b t$, then for all $v\in M^b_2$, we have
$s\rightharpoonup_b v$, by transitivity and connectedness of
$M^b_2$.
\end{itemize}

Let $\mathrm{n}_0,\mathrm{n}_1,\mathrm{n}_2,...,\mathrm{n}_k$ be
all different components of the model $N$. Also suppose
$\mathrm{n}_0$ is a component in which the actual state $s_0$
appears. The model $T'(N)$ is a directed labeled graph in which
the nodes are
$\mathrm{n}_0,\mathrm{n}_1,\mathrm{n}_2,...,\mathrm{n}_k$, and
the edges are agents in $A$. To each node $\mathrm{n}_i$, we
correspond a variable $X_i$.

If the graph is a tree, then we are done and then we can construct
a basic learning program describing $(N,s_0)$. If the graph is
not a tree, we unwind it to an infinite tree with the root $n_0$.
\begin{center}
\scalebox{.6}{\includegraphics{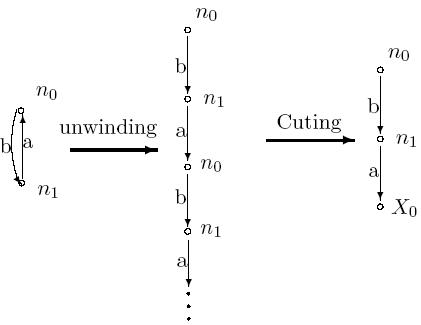}}

Figure~23
\end{center}
In the unwound infinite tree, there could be infinite nodes with
the same name, say $\mathrm{n}_i$. For all nodes $\mathrm{w}$ of
the unwound tree, if $w$ is a node with name $\mathrm{n}_i$ (for
some $i$) and exactly one of its parents has the same name
$\mathrm{n}_i$, then we cut the subtree rooted from $w$ and
change the name of $\mathrm{w}$ from $\mathrm{n}_i$ to variable
$X_i$. In this way, a \emph{finite tree} $T''(N)$ is obtained.

Now we are ready to construct the desired program. We start from
down to the top of the finite tree $T''(N)$.

\begin{itemize}
\item[1.] First note that each leaf of the tree is either a variable
or an $a$-component. If it is  an $a$-component, then we
associate to that leaf, the program
\begin{center}
$L_a(?pre(v_1),?pre(v_2),...,?pre(v_m)$,
\end{center}
where $v_1,v_2,...,v_m$ are all the states of the component. We
note that as the $a$-component is a connected $S5$ action model,
it is associated to the program
\begin{center}
$L_a(?pre(v_1),?pre(v_2),...,pre(v_m)$.
\end{center}
\item[2.] Suppose that $\mathrm{n}_j$ is the name of a node
$\mathrm{w}$ in $T''(N)$, which either all of its children are
corresponded to a variable or a program.  Two cases are possible:
\begin{itemize}
\item Case 1. Among the children of $\mathrm{w}$ there is no node
corresponding to the variable with the same index $j$, that is
$X_j$.

For this case, suppose $\mathrm{n}_j$ refers to a $b$-component
with the states $v_1,v_2,...,v_m$. For each state $v_l$, and each
agent $a\in A$, if there is a directed edge with label $a$
starting from the state $v_l$ to a children of $\mathrm{w}$, say
$\mathrm{u}$, in the tree $T''(N)$, consider
\begin{center}
$pre(v_l)|_a P_{a,l}$
\end{center} where $P_{a,l}$ is a program or variable corresponding to the
node $\mathrm{u}$. Then we associate to the node $\mathrm{w}$, the
program
\begin{center}
$L_b(\bigcap_{a\in A}pre(v_1)|_a P_{a,1},
\bigcap_{a\in A}pre(v_2)|_a P_{a,2},...,\bigcap_{a\in
A}pre(v_m)|_a P_{a,m})$
\end{center}

\item Case 2. Among the children of $\mathrm{n}_j$ there are some
nodes corresponding to the variable with the same index $j$, that
is $X_j$. For this case, we do exactly the same as we did in the
first case, thus obtaining the program
\begin{center}
$L_b(\bigcap_{a\in A}pre(v_1)|_a P_{a,1},
\bigcap_{a\in A}pre(v_2)|_a P_{a,2},...,\bigcap_{a\in
A}pre(v_m)|_a P_{a,m})$.
\end{center}
Then we associate the following program to the node
\begin{center}
$\mu X_j.L_b(\bigcap_{a\in A}pre(v_1)|_a P_{a,1}, \bigcap_{a\in
A}pre(v_2)|_a P_{a,2},...,\bigcap_{a\in A}pre(v_m)|_a P_{a,m})$
\end{center}
\end{itemize}
\end{itemize}
The program corresponding to the root of the tree $T''(N)$ is a
recursive learning program which describes pointed action model
$(N,s_0)$.
\end{proof}

\noindent \begin{proof}\textbf{\ref{main2}.}
In Theorem~\ref{prtree1}, it is proved that the graph of the
action model of a basic learning program is a tree. So none of
the operations: alternative learning, concurrent learning, wrong
learning, produces any loops in the graph of a learning program.
It is easily seen by Definition~\ref{fxp} that, the only operation
that makes loops in the semantics of a learning program is the
recursive learning operator. Therefore, if in a learning program,
there exist $k$ times of dependent use of the recursive operator
$\mu$, then there exists at most a $k$-nested loop in its graph.
That is, for each $k\in \mathbb{N}$, the graph of an action model
associated to a program in $\mathrm{kRLP}$ has at most $k$-nested
loops.

\noindent For each $k>0$, we introduce a learning program
$\alpha^k$, such that its associated action model belongs to
$\mathrm{kRLP}$ but not $\mathrm{(k-1)RLP}$.

\begin{itemize}
\item $k=1$. Let $\alpha^1=\mu X.
L_b(\varphi|_aL_a(\psi|_bL_b(X)))$, where $\varphi$ is not
logically equivalent to $\psi$. The associated action model of the
learning program $\alpha^1$ is the action model $(N_3,s)$ in
Figure~3. Since $\varphi$ and $\psi$ are not logically equivalent,
the two states $(N_3,s)$ and $(N_3,t)$ (see Figure~3) are not
bisimilar.  If there exists a program $\beta$ without any
recursive operator that its associated action model is bisimlar to
$(N_3,s)$, then the action model $N_3$ would be bisimilar to a
finite action model $M'\in FAct$, such that its graph is a tree.
Suppose $R$ is a bisimilarity relation between $N_3$ and $M'$, and
$sRs'$. Because of bisimilarity, since $s\rightharpoonup_a t$,
there exists an state $t'$ in model $M'$, such that $tRt'$ and
$s'\rightharpoonup'_a t'$. Again, since $t\rightharpoonup_b s$,
there exists an state $s''$ in $M'$, such that $s R s''$ and
$t'\rightharpoonup'_b s''$. The model $M'$ is a tree, so we have
$s''\neq s'$, and as $s$ and $t$ are not bisimilar, we have
$s''\neq t'$. Again, by bisimilarity, there exists an state $t''$
in $M'$, such that $s''\rightharpoonup'_a t''$ and $t R t''$. The
new state $t''$ is different from other states of $M'$, since
$M'$ has no loop. In this way, $M'$ is an infinite model, and we
derive a contradiction.

\item $k=2$. The above argument can be done for $k=2$, by
considering the associated action model of the program
$\alpha^2=\mu X.L_a(\varphi|_b\mu Y.
L_b(\psi|_aX\cap\psi|_cL_c(\theta|_bY)))$ (see Figure~17), where
none of the formulas $\varphi$, $\psi$ and $\theta$ are logically
equivalent. If there is a program $\beta$ with at most one use of
recursive operation, then the action model in Figure~17 (which has
a 2-nested loop) would be bisimilar to an action model $M$, that
its graph has just one loop. This can easily be shown, since none
of the states of the action model in Figure~18 are bisimilar to
each other, so the action model $M'$ cannot be finite.

\end{itemize}
So for any arbitrary $k$, we can construct an action model having
one $k$-nested loop, where none of its nodes are bisimilar to each
other. Then this action model is in $\mathrm{kRLP}$ but not
$\mathrm{(k-1)RLP}$.
\end{proof}


\begin{thebibliography}{10}
\bibitem{kn:agm} C. E.  Alchourron, P. Gardenfors, and D. Makinson
{\em  On the logic of theory change: Partial meet contraction and
revision functions}, J. Symb. Log., 50(2):510-530, 1985.

\bibitem{kn:guil} G. Aucher,
{\em Internal models and private multi-agent belief revision},
Proceedings of 7th Int. Conf. on Autonomous Agents and Multi-agent
systems (AAMAS), 2008.

\bibitem{kn:pdl} P. Balbiani, D. Vakarelov,
{\em Iteration-free PDL with intersection of programs: a complete
axiomatization}, Fundamenta Informaticae, Vol 45, Issue 3, 2001.

\bibitem{kn:bal1} A. Baltag, L. Moss, and S. Solecki,
{\em The logic of public announcements, common knowledge and
private Suspicions}, Proceeding TARK 1998, 43-56, Morgan Kaufmann
Publishers, Los Altos, Many update versions.

\bibitem{kn:bal} A. Baltag and L. Moss,
{\em Logics for epistemic programs}, synthese, 139: 165-224, 2004,
Knowledge, Rationality $\&$ Action 1-60.

\bibitem{kn:BV} A. Baltag and H. van Ditmarsch,
{\em Relation between two dynamic epistemic logic}, (abstract),
Proceedings of AAl, 2006.


\bibitem{kn:dit1} H. van Ditmarsch,
{\bf Knowledge Games}, PhD thesis, 2000, ILLC Dissertation Series
DS-2000-06.

\bibitem{kn:dit2} H. van Ditmarsch, W. van der Hoek, and B. Kooi,
{\em Concurrent dynamic epistemic logic}, in V. Hendricks, K.
Jogensen, and S. Pedersen, editors, Knowledge Contributors, page
45-82, Dordercht, 2003, Kluwer Academic Publishers, Synthese
Library Volume 322.

\bibitem{kn:dit3} H. van Ditmarsch, W. van der Hoek, and B. Kooi,
{\bf Dynamic Epistemic Logic}, Springer, 2008.





\bibitem{kn:pen} P. Economou, {\em Sharing beliefs about actions: a
parallel composition operator for epistemic programs}, Summer
school on logic, language and inforamtion, 2005.

\bibitem{kn:wang} J. van Eijck, F. Sietsma, Y. Wang, {\em
Composing Models}, journal of Applied Non-classical Logics,
21(3-4), 2011.

\bibitem{kn:reas} R. Fagin, J. Halpern, Y. Moses, and M. Vardi,
{\bf Reasoning about Knowledge}, MIT Press, Cambridge MA, 1995.

\bibitem{kn:pdl2} D. Harel,
{\em Dynamic Logic}, D. Gabbay, F. Guenthner, Eds., Handbook of
Philosophical logic, vol.II, 1984.



\bibitem{kn:hintk} J. Hintikka, {\bf Knowledge and Belief,
 An Introduction to the Logic of the Two
Notions.} Cornell University Press, Ithaca, New York, 1962.
Republished in 2005 by King�s College, London.



\bibitem{kn:plaz} J. Plaza,
{\em Logics for public communications}, In M. Emrich, M. Pfeifer,
M. Hadzikadic, and Z. Ras, editors, proceedings of the 4th
International  Symposium on Methodologies for Intelligent Systems,
pages 201--216, 1989.

\bibitem{kn:david} D. Steiner,
{\em A system for consistency preserving belief change}, In
Proceedings of Rationality and Knowledge, ed. by S. Artemov and R.
Parikh, pp. 133-144, 2006.
\end{thebibliography}
\end{document}